\documentclass{IEEEtran}
\pdfoutput=1
\usepackage[ascii]{inputenc}
\usepackage[bookmarksnumbered=true,hypertexnames=false,colorlinks=true,linkcolor=blue,urlcolor=blue,citecolor=blue,anchorcolor=green,pdfusetitle]{hyperref}
\usepackage{bbm,braket,microtype,aliascnt,mathrsfs,amsmath,amssymb,amsthm,graphicx,enumitem,mathtools,framed,xspace}
\usepackage[capitalize]{cleveref}
\usepackage[dvipsnames]{xcolor}
\usepackage[T1]{fontenc}
\usepackage[USenglish]{babel}
\newtheorem{thm}{Theorem}\crefname{thm}{Theorem}{Theorems}
\newtheorem{prp}[thm]{Proposition}\crefname{prp}{Proposition}{Propositions}
\newtheorem{lem}[thm]{Lemma}\crefname{lem}{Lemma}{Lemmas}
\crefname{cor}{Corollary}{Corollaries}
\crefname{dfn}{Definition}{Definitions}
\theoremstyle{definition}
\newtheorem*{remark*}{Remark}
\newtheorem*{example*}{Example}

\newcommand{\EE}{\mathbb E}
\newcommand{\N}{\mathbb N}
\newcommand{\R}{\mathbb R}
\newcommand{\ot}{\otimes}
\newcommand{\bigot}{\bigotimes}

\newcommand{\wh}{\widehat}
\newcommand{\id}{\mathbbm 1}

\newcommand{\eps}{\varepsilon}

\newcommand{\state}[1]{\ket{#1}\!\bra{#1}}
\newcommand{\abs}[1]{{\left\vert{#1}\right \vert}}
\newcommand{\norm}[1]{{\left\Vert{#1}\right\Vert}}
\newcommand{\ol}[1]{{\overline{#1}}}
\renewcommand{\O}{\emptyset}
\DeclareMathOperator{\tr}{tr}

\DeclareMathOperator{\ind}{ind}
\DeclareMathOperator{\pa}{pa}
\DeclareMathAlphabet\mathbfcal{OMS}{cmsy}{b}{n}
\usepackage{algorithmic}
\newcommand{\algref}{Algorithm~\hyperlink{alg:codeGen}{1}\xspace}
\allowdisplaybreaks
\begin{document}
\title{A Quantum Multiparty Packing Lemma\\and the Relay Channel}
\author{Dawei Ding, Hrant Gharibyan, Patrick Hayden, and Michael Walter\thanks{Dawei Ding, Hrant Gharibyan, and Patrick Hayden are with the Stanford Institute for Theoretical Physics, Stanford University, Stanford, CA 94305, USA. Michael Walter is with QuSoft, the Korteweg-de Vries Institute for Mathematics, the Institute for Theoretical Physics, and the Institute for Logic, Language and Computation, University of Amsterdam, 1098 XG Amsterdam, The Netherlands. (c) 2018 IEEE. Personal use of this material is permitted.  However, permission to use this material for any other purposes must be obtained from the IEEE by sending a request to pubs-permissions@ieee.org.}}
\date{\today}
\maketitle

\begin{abstract}
  Optimally encoding classical information in a quantum system is one of the oldest and most fundamental challenges of quantum information theory.
  Holevo's bound places a hard upper limit on such encodings, while the Holevo-Schumacher-Westmoreland (HSW) theorem addresses the question of how many classical messages can be ``packed'' into a given quantum system.
  In this article, we use Sen's recent quantum joint typicality results to prove a one-shot multiparty quantum packing lemma generalizing the HSW theorem.
  The lemma is designed to be easily applicable in many network communication scenarios.
  As an illustration, we use it to straightforwardly obtain quantum generalizations of well-known classical coding schemes for the relay channel: multihop, coherent multihop, decode-forward, and partial decode-forward.
  We provide both finite blocklength and asymptotic results, the latter matching existing classical formulas.
  Given the key role of the classical packing lemma in network information theory, our packing lemma should help open the field to direct quantum generalization.
\end{abstract}

\section{Introduction}\label{introduction}
The \emph{packing lemma}~\cite{shannon1948mathematical,cover1975achievable,forney1972information} is one of the central tools used in the construction and analysis of information transmission protocols~\cite{elgamal2011network}.
It quantifies the asymptotic rate at which messages can be ``packed'' reversibly into a medium, in the sense that the probability of a decoding error vanishes in the limit of large blocklength.
For concreteness, consider the following general version of the packing lemma.%
\footnote{See, e.g.,~\cite{elgamal2011network}. Our formulation is slightly paraphrased and uses a notation that is more suitable for the following.}

\begin{lem}[Classical Packing Lemma]\label{lem:cPacking}
Let $(U,X,Y)$ be a triple of random variables with joint distribution $p_{UXY}$.
For each $n$, let $(\tilde U^n,\tilde Y^n)$ be a pair of arbitrarily distributed random sequences
and ${\{ \tilde X^n(m) \}}$ a family of at most $2^{nR}$ random sequences such that each $\tilde X^n(m)$ is conditionally independent of $\tilde Y^n$ given $\tilde U^n$ (but arbitrarily dependent on the other $\tilde X^n(m')$ sequences).
Further assume that each $\tilde X^n(m)$ is distributed as $\otimes_{i=1}^n p_{X|U=\tilde U_i}$ given $\tilde U^n$.
Then, there exists $\delta(\eps)$ that tends to zero as $\eps\to0$ such that
\begin{align*}
 \lim_{n\to\infty} \Pr((\tilde U^n,\tilde X^n(m),\tilde Y^n)\in\mathcal T_\eps^{(n)} \text{ for some $m$}) = 0
\end{align*}
if $R < I(X;Y|U)-\delta(\eps)$, where $\mathcal T_\eps^{(n)}$ is the set of $\eps$-typical strings of length~$n$ with respect to $p_{UXY}$.
\end{lem}

The packing lemma provides a unified approach to many, if not most, of the achievability results in Shannon theory.
Despite its broad utility, it is a simple consequence of the union bound and the standard joint typicality lemma with the three variables $U$, $X$, $Y$.
The usual channel coding theorem directly follows from taking $U = \O$ and when $\tilde Y^n \sim p_Y^{\ot n}$.

For the case when $U = \O$ and when $\tilde Y^n \sim p_Y^{\ot n}$, the quantum generalization of the packing lemma is known: the Holevo-Schumacher-Westmoreland (HSW) theorem~\cite{holevo1998capacity,schumacher1997sending}.
This can be proven using a conditional typicality lemma for a classical-quantum state with one classical and one quantum system.
However, until recently no such typicality lemma was known for the case of multiple encoding systems, and so a quantum version of~\cref{lem:cPacking} was lacking.
Furthermore, while in classical Shannon theory~\cref{lem:cPacking} can be used repeatedly in scenarios where the message is encoded into multiple random variables, this approach fails in the quantum case due to measurement disturbance, specifically the influence of one decoding on subsequent decodings.
Hence, while it is sufficient to solve the full multiparty packing problem in the classical case with just two encoding systems and repeated measurements, a general multiparty packing lemma with $k \in \N$ encoding systems is required in the quantum case.
The bottleneck is again the lack of a general quantum joint typicality lemma with multiple systems.
However, we can obtain partial results in the quantum case for some network scenarios, as we will describe below.

In this paper we use the quantum joint typicality lemma%
\footnote{Sen modestly calls his result a lemma, but the highly ingenious proof more than justifies calling it a theorem.}
established recently by Sen~\cite{senInPrep} to prove a quantum one-shot multiparty packing lemma for $k$ classical encoding systems.
We then demonstrate the wide applicability of the lemma by using it to
generalize classical network information theory protocols
to the quantum case.
The lemma allows us to construct and prove the correctness of these simple generalizations and, we believe, should help to open the field of classical network information theory to direct quantum generalization.%
\footnote{Note that a simultaneous smoothing result for the max-relative entropy is still missing, which would be necessary e.g.~for a ``multiparty covering lemma.'' To prove such a result is a major open problem in the field.}
One feature of the lemma is that it leads naturally to demonstrations of the achievability of rate regions without having to resort to time-sharing, a desirable property known as \emph{simultaneous decoding}.
Simultaneous decoding
is often necessary in network information theory to obtain one-shot rates for the full achievable rate region. This region is often a convex closure of the union of different regions, where convex combinations of rates are usually achieved through time-sharing. This is not possible in a one-shot setting. Furthermore,
different receivers could have different effective rate regions and therefore require incompatible time-sharing strategies.
Indeed, this is a frequent source of incomplete or incorrect results even in classical information theory~\cite{DBLP:journals/corr/abs-1207-0543}.
A general construction leading to simultaneous decoding in the quantum setting has therefore been sought for many years~\cite{DBLP:journals/corr/abs-1207-0543,fawzi2012classical,dutil2011multiparty,winter2001capacity,fawzi2011quantum,christandl2018recoupling,walter2014multipartite}.
Sen's quantum joint typicality lemma achieves this goal, as does our packing lemma, which can be viewed as a user-friendly interface for Sen's lemma.

Recall that network information theory is the study of communication
with multiple parties and is a generalization of the conventional single-sender single-receiver two-party scenario, commonly known as \emph{point-to-point} communication.
Common network scenarios include having multiple senders encoding different messages, as in the case of the multiple access channel~\cite{shannon1961two}, multiple receivers decoding the messages, as in the broadcast channel~\cite{cover1972broadcast}, or a combination of both, as for the interference channel~\cite{ahlswede1974capacity}.
However, the above examples are all instances of what is called \emph{single hop} communication, where the message directly travels from a sender to a receiver.
In \emph{multihop} communication, there is one or even multiple intermediate nodes where the message is decoded or partially decoded before being transmitted to the final receiver.
Examples of such communication scenarios include the relay channel~\cite{van1971three}, which we focus on in this paper, and more generally, graphical multi-cast networks~\cite{kramer2005cooperative,xie2005achievable}.

Research in quantum joint typicality has generally been driven by the need to establish quantum generalizations of results in classical network information theory. Examples include the quantum multiple access channel~\cite{winter2001capacity,yard2008capacity}, the quantum broadcast channel~\cite{yard2011quantum,dupuis2010father}, and the quantum interference channel~\cite{fawzi2011quantum}. Indeed, some partial results on joint typicality had been established or conjectured in order to prove achievability bounds for various network information processing tasks~\cite{dutil2011multiparty,sen2012achieving,qi2018applications}. Subsequent work made some headway on the abstract problem of joint typicality for quantum states, but not enough to affect coding theorems~\cite{drescher2013simultaneous,notzel2012solution} prior to Sen's breakthrough~\cite{senInPrep}.

The quantum relay channel was studied previously in~\cite{savov2012partial}, where the authors constructed a partial decode-forward protocol. Here we develop finite blocklength results for the relay channel in addition to reproducing the earlier conclusions and avoiding a resolvable issue with error accumulation from successive measurements in their partial decode-forward bound. (We construct a joint decoder which obtains all the messages from the multiple rounds of communication simultaneously.) Our analysis makes extensive use of our quantum multiparty packing lemma. Once the coding strategy is specified, a direct application of the packing lemma in the asymptotic limit gives a list of inequalities which describe the rate region, which we then simplify using entropy inequalities to the usual rate region of the partial decode-forward lower bound. There has also been related work in~\cite{jin2012lower}, which considered concatenated channels, a special case of the more general relay channel model. As noted in~\cite{savov2012partial}, work on quantum relay channels may have applications to designing quantum repeaters~\cite{collins2005quantum}. Note that Sen has already used his joint typicality lemma to prove achievability results for the quantum multiple access, broadcast, and interference channels~\cite{senInPrep, senInPrep2}, but here we give a general packing lemma which can be used as a black box for quantum network information applications. The relay channel serves as a demonstration of this.

Our paper is structured as follows.
In~\cref{sec:prelim}, we establish notation and discuss some preliminaries.
In~\cref{sec:packinglemma}, we describe the setting and state the quantum multiparty packing lemma.
The statement very much resembles a one-shot, multiparty generalization of~\cref{lem:cPacking}, but, to reiterate, while the multiparty generalization is trivial in the classical case, it requires the power of a full joint typicality lemma in the quantum case.
In~\cref{sec:quantumrelay} we describe the classical-quantum (c-q) relay channel and systematically describe coding schemes that generalize known schemes
for the classical relay channel: multihop, coherent multihop, decode-forward, and partial decode-forward~\cite{cover1979capacity}.
In addition to the one-shot bounds, we show that the asymptotic bounds are obtained by taking the limit of large blocklength, thereby obtaining quantum generalizations of known capacity lower bounds for the classical case.
In~\cref{sec:proofpackinglemma} we prove the quantum multiparty packing lemma via Sen's quantum joint typicality lemma~\cite{senInPrep}.
For convenience, we restate a special case of the Sen's joint typicality lemma and suppress some of the details.
In~\cref{sec:conclusions} we give a conclusion.

\section{Preliminaries}\label{sec:prelim}
We first establish some notation and recall some basic results.

\smallskip\noindent\textbf{Classical and quantum systems:}
A classical system $X$ is identified with an alphabet $\mathcal X$ and a Hilbert space of dimension $\abs{\mathcal{X}}$, while a quantum system $B$ is given by a Hilbert space of dimension~$d_B$.
Classical states are modeled by diagonal density operators such as $\rho_X=\sum_{x\in\mathcal X} p_X(x)\ket x\bra x_X$, where $p_X$ is a probability distributions, quantum states are described by density operator $\rho_A$ etc, and classical-quantum states are described by density operators of the form
\begin{align}\label{eq:cq state}
  \rho_{X B} =  \sum_{x \in \mathcal X} p_{X}(x) \state{x}_{X} \ot \rho_B^{(x)}.
\end{align}

\noindent\textbf{Probability bound:}
Denote by $E_1$, $E_2$ two events. We use the following inequality repeatedly in the paper:
\begin{align}
\nonumber
\Pr(E_1)
&= \Pr(E_1|E_2)\Pr(E_2) + \Pr(E_1|\ol E_2)\Pr(\ol E_2) \\
\label{eq:easybound}
&\leq  \Pr(E_2) + \Pr(E_1|\ol E_2),
\end{align}
where we use $\ol{E_2}$ to denote the complement of $E_2$ and used the fact that $\Pr(E_2), \Pr(E_1|\ol{E_2}) \leq 1$.

\noindent\textbf{Hypothesis-testing relative entropy:}
The hypothesis-testing relative entropy~\cite{hirche2018thesis} is defined as%
\footnote{As always in information theory, $\log$ here is base 2.}
\begin{align*}
  D_H^\eps(\rho \Vert \sigma) = \max_{\substack{ 0 \le \Pi \le I\\ \tr(\Pi \rho) \geq 1-\eps }} - \log \tr(\Pi \sigma).
\end{align*}
For $n$ copies of states $\rho$ and $\sigma$,~\cite{datta2011strong,tomamichel2013hierarchy,li2014second} establishes the following inequalities:
\begin{align}\label{eq:DHepsBound}
  D(\rho \Vert \sigma) - \frac{F_1(\eps)}{\sqrt{n}} \leq \frac{1}{n} D_H^\eps(\rho^{\otimes n} \Vert \sigma^{\otimes n}) \leq D(\rho \Vert \sigma) + \frac{F_2(\eps)}{\sqrt{n}},
\end{align}
where $F_1(\eps) , F_2(\eps) \geq 0$ are given by $F_1(\eps) \equiv 4\sqrt{2} \log \frac{1}{\eps} \log \eta$, $F_2(\eps) \equiv 4\sqrt{2} \log \frac{1}{1-\eps} \log \eta$, with $\eta \equiv 1+\tr \rho^{3/2}\sigma^{-1/2} +\tr \rho^{1/2}\sigma^{1/2}$ and $D(\rho \Vert \sigma) = \tr (\rho \log \rho) -  \tr (\rho \log \sigma)$ being the quantum relative entropy.
In the limit of large $n$, we obtain the quantum Stein's lemma~\cite{ogawa2005strong,hiai1991proper}:
\begin{align}\label{eq:stein}
  \lim_{n \to \infty} \frac{1}{n} D_H^\eps(\rho^{\otimes n} \Vert \sigma^{\otimes n}) = D(\rho \Vert \sigma).
\end{align}

\noindent\textbf{Conditional density operators:}
Let a classical system $X$ consist of subsystems $X_v$, for $v$ in some index set~$V$, with alphabet $\mathcal X = \bigtimes_{v\in V} \mathcal X_v$, where $\times$ denotes the Cartesian product of sets.
Consider a classical-quantum state $\rho_{XB}$ as in~\cref{eq:cq state} and a subset $S \subseteq  V$.
We can write
\begin{align}
  \label{eq:decompCond}
  \rho_{X B} =  \sum_{x_\ol{S}}^{} p_{X_{\ol{S}}}(x_{\ol{S}}) \state{x_\ol{S}}_{X_{\ol{S}}} \ot \rho_{X_S B}^{(x_{\ol{S}})},
\end{align}
where $\ol{S} \equiv V \setminus S$ and
\begin{align*}
  \rho^{(x_{\ol{S}})}_{X_S B} \equiv \sum_{x_S} p_{X_S|X_\ol{S}} (x_{S}|x_{\ol{S}})  \state{x_S}_{X_S} \otimes \rho^{(x_S, x_{\ol{S}})}_{B}.
\end{align*}
We can interpret $\rho^{(x_{\ol{S}})}_{X_S B}$ as a ``conditional'' density operator.
We further define $\rho_{X B}^{(\left\{ X_S, B\right\})}$ by replacing the conditional density operator in~\cref{eq:decompCond} by the tensor product of its marginals:
\begin{align*}
  \rho_{X B}^{(\left\{ X_S,B \right\})} & = \sum_{x_{\ol{S}}} p_{X_\ol{S}} (x_{\ol{S}})  \state{x_{\ol{S}}}_{X_\ol{S}} \otimes \rho^{(x_{\ol{S}})}_{X_S} \otimes \rho^{(x_{\ol{S}})}_{B}  \\
  & = \sum_{x} p_{X}(x)  \state{x}_{X} \otimes \rho^{(x_{\ol{S}})}_{B}.
\end{align*}
This formulation lets us obtain the conditional mutual information as an asymptotic limit of the hypothesis testing relative entropy; by~\cref{eq:stein},
\begin{align}
  & \lim_{n\to \infty} \frac{1}{n} D_H^\eps \Big(\rho_{X B}^{\ot n} \Vert \left( \rho_{X B}^{(\{X_S, B \})} \right)^{\ot n}\Big) \nonumber\\
  & = D(\rho_{X B} \Vert \rho_{X B}^{(\{X_S, B \})}) \nonumber \\
  & = \sum_{x_{\ol{S}}} p_{X_\ol{S}}(x_{\ol{S}}) D\Big(\rho^{(x_{\ol{S}})}_{X_{S} B} \;\Vert\; \rho^{(x_{\ol{S}})}_{X_S} \ot \rho^{(x_{\ol{S}})}_{B}\Big) \nonumber\\
  & = \sum_{x_{\ol{S}}}^{} p_{X_\ol{S}}(x_{\ol{S}}) I(X_S; B)_{\rho^{(x_\ol{S})}} \nonumber \\
  & = I(X_S; B | X_\ol{S})_{\rho}, \label{eq:CMI}.
\end{align}

\section{Quantum Multiparty Packing Lemma}\label{sec:packinglemma}
In this section, we formulate a general multiparty packing lemma for quantum Shannon theory that can be used as a black box for network coding constructions. The goal is to ``pack'' as many classical messages as possible into our quantum system while retaining distinguishability. In the multiparty case, we are packing classical messages via an encoding that involves \emph{multiple} classical systems. As mentioned in the introduction, a multiparty packing lemma is necessary in quantum information theory due to measurement disturbance. That is, while in classical information theory one can do consecutive decoding operations on the same quantum system with impunity, in quantum information theory a decoding operation can change the system and thereby affect a subsequent operation. For example, while classically it is possible to check whether the output of a channel is typical with multiple input random variables by simply verifying typicality pair by pair, quantumly this method can be problematic. Hence, we would like to combine a set of decoding operations into one, \emph{simultaneous} decoding.
We obtain a construction of this flavor in~\cref{lem:packing}.
Its asymptotic version,~\cref{lem:packingAsympt}, states that the decoding error vanishes provided that a \emph{set} of inequalities on the rate of transmission is satisfied, as opposed to a single one as in~\cref{lem:cPacking}.
This is exactly what we expect from a simultaneous decoding operation.

We first need to establish what it means to have a ``multiparty'' packing lemma. In network information theory scenarios, it is often necessary to have multiple message sets, representing in the simplest cases transmissions to and from different users or in different rounds of communication. Random codewords may be generated for each message, but the dependence of the codewords on the different message sets may be complicated. Furthermore, the codewords may be correlated in intricate ways. 
In order to demonstrate these concepts and to motivate the formal statements to come, it is helpful to have an example in mind. The example we will use is the two-sender c-q multiple access channel (MAC), for which asymptotic rates were obtained in~\cite{winter2001capacity} and one-shot rates in~\cite{senInPrep}. This channel is simply a c-q channel with two classical inputs: $X_1, X_2$, one for each sender. We also have two message sets, $M_1, M_2$, corresponding to the messages the two senders wish to transmit. A random coding scheme can be used where a codeword $x_1(m_1)$ is randomly generated according to a probability distribution $p_{X_1}$ for every $m_1 \in M_1$ and similarly for the second sender. In order to obtain a simultaneous decoder for the full rate region, including interpolations between different probability distributions $p_{X_1} p_{X_2}$, we introduce the time-sharing variable $U$. The codeword for $U$ is randomly generated according to some probability distribution $p_U$ which determines the precise interpolation. The codewords for $X_1, X_2$ are generated conditioned on $U$. For more details, see for instance the full classical treatment in~\cite{elgamal2011network}.

This encoding scheme can be represented graphically by a mathematical object that we call a \emph{multiplex Bayesian network} (\cref{fig:relayexample}, explained below). This object is key to the technical setup of our multiparty packing lemma. Note that graphical constructions for network coding is not a new concept (see for instance, Section A.9 in~\cite{kramer2008topics}). Our construction in particular can be interpreted as a mathematical formalization of Markov encoding schemes, which are ubiquitous in network information theory~\cite{elgamal2011network}.

We now give the mathematical description of a multiplex Bayesian network. Let the joint random variable $X$ be a Bayesian network with respect to a directed acyclic graph (DAG) $G = (V,E)$.
The joint random variable $X$ is composed of random variables $X_v$ with alphabet $\mathcal{X}_v$ for each $v \in V$.
Now, a multiplex Bayesian network graphically represents a random coding scheme via an algorithm we give below which takes the multiplex Bayesian network as input and generates random codewords $x(m)$ with components $x_v(m)$ for $v \in V$.
However, different components of a codeword may only depend on particular message sets, as in the case of the MAC where we only generate $x_1(m_1)$ for all $m_1 \in M_1$.
We model this situation by an index set $J$ which index the different message sets $M_j$ for $j\in J$, and a function $\ind: V \to \mathcal{P}(J)$, where $\ind(v)\subseteq J$ corresponds to the subset of indices (hence the name) that index the message sets the codeword component $X_v$ depends on.
Now, for our random codebook construction to be well-defined, we require that given $v \in V$,
\begin{align}\label{eq:ind condition}
  \ind(v') \subseteq \ind(v) \text{ for every } v' \in \pa(v),
\end{align}
where for $v\in V$,
\begin{align*}
  \pa(v) \equiv \left\{ v' \in V \;\vert\; (v', v) \in E \right\}
\end{align*}
denote the set of parents of $v$. This is a natural requirement: in our algorithm we generate the codewords in an iterative manner following the edges in the DAG, and so we should require that all codeword components also depend on the message sets that the components upstream depend on. That is, codeword components ``inherit'' the indices of their parents.


We call the tuple $\mathcal{B} = (G, X, M, \ind)$, where $M \equiv \bigtimes_{j \in J} M_j$, a \emph{multiplex Bayesian network}.
We can visualize a multiplex Bayesian network by adjoining to the DAG $G$ additional vertices $M_j$, one for each $j\in J$, and edges that connect each $X_v$ to every $M_j$ such that $j\in\ind(v)$. Again, as an example consider the MAC multiplex Bayesian network with three random variables in~\cref{fig:relayexample}. Note that in the figure the random variable $U$ is not connected with any message set. We define our algorithm to treat $U$ as if it is connected with a singleton.

\begin{figure}[th]
\begin{center}
\includegraphics[scale=0.7]{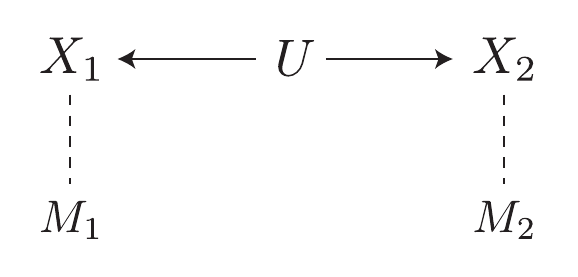}
\end{center}
\caption{An example of a multiplex Bayesian network with vertices $U$, $X_1$, and $X_2$ and message sets $M_1$ and $M_2$. This network can be used to generate a random code for the two-sender c-q MAC, where we generate $u$ according to $p_U$, then $x_1(m_1)$ for $m_1 \in M_1$ according to $p_{X_1\vert U}(\cdot \vert u)$ and $x_2(m_2)$ according to $p_{X_2 \vert U}(\cdot \vert u)$. }\label{fig:relayexample}
\end{figure}

We next give the algorithm that generates the random codebook. Given a multiplex Bayesian network $\mathcal{B} = (G, X, M, \ind)$, we would like to generate a random codebook
\begin{align}
  \label{eq:codebook}
  \left\{ x_v(m) \right\}_{ v \in V, m \in M},
\end{align}
where $x_v$ is a random variable with alphabet $\mathcal{X}_v$.
The vertices represent random codeword components and the graph $G$
describes the dependencies between different components.
Moreover, each component $x_v(m)$ only depends on $m_j\in M_j$ for which $j\in \ind(v)$.
That is, $x_v(m)$ and $x_v(m')$ are equal as random variables provided $m_j=m'_j$ for every $j\in\ind(v)$.

We now give the algorithm for generating the random codebook.
Since $G$ is a DAG, it has a topological ordering, that is, a total ordering on $V$ such that for every $(v',v) \in E$, $v'$ precedes $v$ in the ordering.
We also pick an arbitrary total ordering on $J$ and on $M_j$ for every $j \in J$.
This then induces a lexicographical ordering on Cartesian products of $M_j$, which we denote by $M_{J'} := \bigtimes_{j\in J'} M_j$ for any $J' \subseteq J$.
We define $M_\emptyset=\{\emptyset\}$ as a singleton so that we can identify $M_{J'} \times M_{J''} = M_{J' \cup J''}$ for any two disjoint subsets $J',J''\subseteq J$. 
Note that these total orderings determine the order in which we perform the for loops below, but do not impact the joint distribution of the codewords. We can now define the following algorithm:
\begin{framed}
\noindent\hypertarget{alg:codeGen}{\textbf{Algorithm 1:} Codebook generation from multiplex Bayesian network}\\ 
\begin{algorithmic}
  \FOR{$v \in V$}
    \FOR{$m_v \in M_{\ind(v)}$}
      \STATE generate $x_v(m_v)$ according to $p_{X_v | X_{\pa(v)}}\left(\cdot | x_{\pa(v)}(m_{\pa(v)})\right)$
      \FOR{$m_{\bar v} \in M_{\ol{\ind(v)}}$}
	\STATE $x_v(m_v, m_{\bar v}) = x_v(m_v)$
      \ENDFOR
    \ENDFOR
  \ENDFOR
\end{algorithmic}
\end{framed}
Here, $\ol{\ind(v)} \equiv J \setminus \ind(v)$, $m_{\pa(v)}$ is the restriction of $m_v$ to $M_{\ind(\pa(v))}$ (this makes sense by~\cref{eq:ind condition}), $X_{\pa(v)} \equiv (X_{v'})_{v'\in\pa(v)}$ and similarly for $x_{\pa(v)}(m_{\pa(v)})$, and the pair $(m_v, m_{\bar v})$ is interpreted as an element of $M$ with the appropriate components.
The topological ordering on $V$ ensures that~$x_{\pa(v)}(m_{\pa(v)})$ is generated before $x_v(m_v)$, so this algorithm can be run.
We thus obtain a random codebook as in~\cref{eq:codebook}.

We make a few observations.
\begin{enumerate}
\item\label{it:obs1}
By construction, for all $m \in M$ and $\xi\in\mathcal X$,
\begin{align*}
  \Pr(x(m) = \xi) = p_X(\xi)
\equiv \prod_{v \in V} p_{X_v | X_{\pa(v)}}\left(\xi_v | \xi_{\pa(v)}\right).
\end{align*}
That is, $x(m)$ is a Bayesian network with respect to $G$ and equal in distribution to $X$.
\item\label{it:obs2}
By construction, given $v \in V$ and $m_v\in M_{\ind(v)}$, all $x_v(m_v, m_{\bar v})$ for $m_{\bar v} \in M_{\ol{\ind(v)}}$ are equal as random variables.
\item\label{it:obs3}
Generalizing observation~\ref{it:obs1}, the joint distribution of \emph{all} codewords can be split into factors in a simple manner.
Specifically, given $\xi(m)\in\mathcal X$ for every $m\in M$, we have
\begin{align*}
  & \Pr(x(m) = \xi(m) \text{ for all } m\in M)
\\
& = \prod_{v\in V} \prod_{m_v\in M_{\ind(v)}} p_{X_v|X_{\pa(v)}}\left(\xi_v(m_v)|\xi_{\pa(v)}(m_{\pa(v)})\right)
\end{align*}
provided $\xi_v(m) = \xi_v(m')$ for all $m, m'$ with $m_v = m'_v$.
Otherwise, the joint probability is zero.
\end{enumerate}

To allow for the variety of coding schemes encountered in network information theory, we introduce a few additional elements. For instance, we would like the freedom to construct multiple different quantum decoders for the same random codebook. This is a very natural requirement when there are multiple receivers involved or when a particular receiver has to make multiple measurements in an interactive communication scenario (and therefore cannot simply make a single joint measurement).
To realize this, let $H$ be the induced subgraph of $G$ for some $V_H \subseteq V$ where for all $v \in V_H$, $\pa(v) \subseteq V_H$.
We call $H$ an \emph{ancestral subgraph}.
Then, we can naturally define $X_H$ to be the set of random variables corresponding to $V_H$, $J_H \equiv \bigcup_{v \in V_H} \ind(v) \subseteq J$, $M_H \equiv \bigtimes_{j \in J_H} M_j$, and $C_H \equiv \left\{ x_H(m_H) \right\}_{m_H \in M_H}$.%
\footnote{Note that by the definition of $M_H$ we only need $m_H$ to identify $x_H$ up to equality as random variables.}
Finally, the classical variables are encoded into a quantum system via a family of quantum states $\left\{ \rho_{B}^{(x_H)} \right\}_{x_H \in \mathcal{X}_H}$, where $B$ is some quantum system.

The next element we introduce allows for receivers to decode a number of message sets using a \emph{guess} for the other message sets. This is naturally motivated by iterative decoding schemes in which a receiver makes multiple measurements where latter measurements take into account results from previous measurements. Again, this is mainly relevant in interactive scenarios.
To realize this, let $D \subseteq J_H$ be a subset of indices which index the message sets to be decoded. This means we have a guess for the remaining message sets indexed by $\ol D \equiv J_H \setminus D$.

We can now state our quantum multiparty packing lemma:

\begin{lem}[One-shot quantum multiparty packing lemma]\label{lem:packing}
  Let $\mathcal{B} = (G, X, M, \ind)$ be a multiplex Bayesian network and run \algref to obtain a random codebook $C = \left\{ x(m) \right\}_{m \in M}$.
  Let $H \subseteq G$ be an ancestral subgraph, $\{ \rho_{B}^{(x_H)} \}_{x_H \in \mathcal{X_H}}$ a family of quantum states, $D \subseteq J_H$, and $\eps\in(0,1)$.
  Then there exists a POVM%
\footnote{These POVMs depend on the codebook $C_H$ and are hence involved in the averaging in~\cref{eq:packing}. This will be important in the analyses below.}
  $\{ Q_B^{(m_D \vert m_{\ol D})} \}_{m_D \in M_D}$
  for each $m_{\ol D}\in M_{\ol D}$
  such that, for all $(m_D,m_{\ol D}) \in M_H$,
  \begin{align}
    & \EE_{C_H}\left[ \tr\left[(I - Q_B^{(m_D | m_{\ol D})}) \rho^{(x_H(m_D, m_{\ol D}))}_B \right] \right]  \nonumber\\
    & \leq f(\abs{V_H}, \varepsilon) + 4 \sum_{\O \neq T \subseteq D}  2^{\scalebox{0.7}{$\displaystyle \big(\sum_{t \in T}^{} R_t \big) - D_H^\epsilon(\rho_{X_H B} \Vert \rho^{(\{X_{S_T}, B \})}_{X_H B} )$}}. \label{eq:packing}
  \end{align}
  Here, $\EE_{C_H}$ denotes the expectation over the random codebook $C_H=\{x_H(m_H)\}_{m_H\in M_H}$,
  $R_t \equiv \log \abs{M_t}$,
  \begin{align*}
    S_T \equiv \left\{ v \in V_H \;\vert\; \ind(v) \cap T \neq \O \right\},
  \end{align*}
  and
  \begin{align*}
  \rho_{X_H B} \equiv \sum_{x_H \in \mathcal{X}_H} p_{X_H}(x_H) \state{x_H}_{X_H} \otimes \rho^{(x_H)}_B.
  \end{align*}
  Furthermore, $f(k, \eps)$ is a universal function (independent of our setup) that tends to zero as $\eps\to0$.
\end{lem}
\newpage

\begin{remark*}
  The bound in~\cref{eq:packing} can also be written as
  \begin{align}\label{eq:packing2}
    & \EE_{C_H}\left[ \tr\left[(I - Q_B^{(m_D | m_{\ol D})}) \rho^{(x_H(m_D, m_{\ol D}))}_B \right] \right] \nonumber\\
    & \leq f(\abs{V_H}, \varepsilon) +
    4 \sum_{m_D' \neq m_D}^{} 2^{-D_H^\varepsilon(\rho_{X_H B} \Vert \rho_{X_H B}^{(\{X_S, B \})})},
  \end{align}
  where
  \begin{align*}
    S \equiv \left\{ v \in V_H \;\vert\; \exists j \in D \cap \ind(v) \text{ such that } (m_D)_j \neq (m'_D)_j \right\}.
  \end{align*}
  In words, $S$ is the set of random codewords that depend on a part of the message that differs between $m_D$ and $m_D'$.
  This is similar to decoding error bounds obtained with conventional methods, such as the Hayashi-Nagaoka lemma~\cite{hayashi2003general}.
  We obtain~\cref{eq:packing} from~\cref{eq:packing2} by parametrizing the different $m_D'$ with respect to the indices that differ from $m_D$.
\end{remark*}

\begin{remark*}
  Note~\cref{eq:packing} assumes that the decoder's guess of $m_{\overline D}$ is correct. That is, they choose the POVM $\left\{ Q_B^{(m_D | m_{\ol D})} \right\}_{m_D \in D}$, where $m_{\ol D}$ is exactly the $m_{\ol D}$ in the encoded state $\rho^{(x_H(m_D, m_{\ol D}))}_B$. 
  If the decoder's guess is incorrect, then this bound does hold in general. In applications, $m_{\ol D}$ typically corresponds to message estimates of previous rounds, which we will assume to be correct by invoking a classical union bound. That is, we bound the total probability of error by summing the probabilities of error of a decoding \emph{assuming} that all previous decodings were correct. Note that the decodings must be performed on disjoint quantum systems for this argument to hold.
\end{remark*}

The following is the explicit form of $f(k, \varepsilon)$ for $k \in \N$ from~\cite{senInPrep} and our proof of the packing lemma in~\cref{sec:proofpackinglemma}:
\begin{align}
  f(k, \varepsilon)   & = \left( 1+ 6 \times 2^{\frac{k+1}{2}} + 4 \times 2^{2^{k+5} + k^2 + 2k} \right) \varepsilon^{1/3}.
  \label{eq:explicitF}
\end{align}
For simplicity, we can make some coarse approximations to obtain an upper bound:
\begin{align}
  \label{eq:upperBoundF}
  f(k,\varepsilon) \le 2^{2^{7k}} \varepsilon^{1/3}.
\end{align}

Using~\cref{lem:packing} and~\cref{eq:CMI}, we can naturally obtain the asymptotic version where we simply take $n \in \N$ copies of the codebook and take the limit of large $n$. By the quantum Stein's lemma~\cref{eq:stein}, the error in~\cref{eq:packing} will vanish if the rates of encoding are bounded by conditional mutual information quantities. We present this as a self-contained statement.
\begin{lem}[Asymptotic quantum multiparty packing lemma]\label{lem:packingAsympt}
  Let $\mathcal{B} = (G, X, M, \ind)$ be a multiplex Bayesian network.
  Run \algref $n$ times to obtain a random codebook $C^n = \left\{  x^n(m) \in \mathcal X^n\right\}_{m \in M}$.
  Let $H \subseteq G$ be an ancestral subgraph, $\{ \rho_{B}^{(x_H)} \}_{x_H \in \mathcal{X}_H}$ a family of quantum states, and $D \subseteq J_H$.
  Then there exists a POVM $\{ Q_{B^n}^{(m_D \vert m_{\ol D})} \}_{m_D \in M_D}$ for each $m_{\ol D} \in M_{\ol D}$ such that, for all $(m_D,m_{\ol D}) \in M_H$,
  \begin{align*}
    \lim_{n \to \infty} \EE_{C_H^n} \left[ \tr \left[ (I - Q_{B^n}^{(m_D | m_{\ol D})}) \bigot_{i=1}^n \rho_{B_i}^{(x_{i,H}(m_D, m_{\ol D}))} \right] \right] = 0,
  \end{align*}
  provided that%
\footnote{Note that $R_t$ is defined differently here. This is due to the difference in the definition of ``rate'' for one-shot and asymptotic settings. }
  \begin{align*}
    \sum_{t \in T}^{} R_t < I(X_{S_T} ; B | X_{\overline {S_T}})_\rho -\delta(n)
    \quad\text{ for all } \emptyset\neq T \subseteq D.
  \end{align*}
  Above, $\EE_{C_H^n}$ is the expectation over the random codebook $C_H^n \equiv \left\{ x_H^n(m_H) \right\}_{m_H \in M_H}$, $R_t \equiv \frac{1}{n} \log \abs{M_t}$, $\delta(n)$ is some function that tends to $0$ as $n \to \infty$,
  \begin{align*}
    S_T \equiv \left\{ v \in V_H \;\vert\; \ind(v) \cap T \neq \O \right\}, \, \ol{S_T} \equiv V_H \setminus S_T,
  \end{align*}
  and
  \begin{align*}
    \rho_{X_H B} \equiv \sum_{x_H \in \mathcal{X}_H} p_{X_H}(x_H) \state{x_H}_{X_H} \otimes \rho^{(x_H)}_B.
  \end{align*}
\end{lem}

\begin{example*}
  \label{ex:packing}
  To clarify the definitions and illustrate the applications of~\cref{lem:packing} and~\cref{lem:packingAsympt}, we use them to code over the two-sender c-q MAC.
  Consider the multiplex Bayesian network given in~\cref{fig:relayexample}. We apply~\algref to obtain a random codebook $\left\{ u, x_1(m_1), x_2(m_2) \right\}_{m_1 \in M_1, m_2 \in M_2}$. We then simply let each sender transmit their message via the corresponding codeword. Now, choosing%
\footnote{We introduced these elements mainly for interactive scenarios, such as the relay channel that we analyze below. }
  $H=G$ 
  and $D=J = \left\{ 1,2 \right\}$, by~\cref{lem:packingAsympt} we obtain a POVM $\{ Q_B^{(m_1, m_2)} \}_{m_1 \in M_1, m_2 \in M_2}$. The mapping from $T \subseteq D$ to $S_T \subseteq V = \left\{ U, X_1, X_2 \right\}$ is given in~\cref{tab:ST}.
  \begin{table}[h]
  \renewcommand{\arraystretch}{1.5}
  \caption{$S_T \subseteq V$ for various $\O \neq T \subseteq D$.}
  \label{tab:ST}
  \centering
  \begin{tabular}{c | c}
    $T$ & $S_T$\\
    \hline
    $\left\{ 1 \right\}$ & $\left\{ X_1 \right\}$\\
    $\left\{ 2 \right\}$ & $\left\{ X_2 \right\}$\\
    $\left\{ 1,2 \right\}$ & $\left\{ X_1, X_2 \right\}$
  \end{tabular}
  \end{table}
  \\
  Thus, letting the receiver use this POVM achieves the rate region
  \begin{align*}
    R_1 & < I(X_1; B \vert X_2 U)_\rho \\
    R_2 & < I(X_2; B \vert X_1 U)_\rho \\
    R_1 + R_2 & < I(X_1 X_2; B \vert U)_\rho,
  \end{align*}
  where
  \begin{align*}
    \rho_{U X_1 X_2 B} = \sum_{u, x_1, x_2}^{} & p_{U}(u) p_{X_1 \vert U}(x_1 \vert u) p_{X_2 \vert U}(x_2 \vert u) \\
    & \state{u,x_1, x_2}_{U X_1 X_2} \ot \rho_B^{(x_1, x_2)}
  \end{align*}
  and $\rho_B^{(x_1, x_2)}$ is the output of the MAC with input $(x_1, x_2)$. Hence, with our quantum multiparty packing lemma we readily achieve the capacity found in~\cite{winter2001capacity}.

  We can also get one-shot results for the MAC. Let $R_1, R_2, \varepsilon \in \R_{\ge 0}$ such that
  \begin{align*}
    R_1 &\le D_H^\varepsilon(\rho_{U X_1 X_2 B} \Vert \rho_{U X_1 X_2 B}^{( \left\{ X_1, B \right\} )}) - 2 - \log \frac{1}{\varepsilon} \\
    R_2 & \le D_H^\varepsilon(\rho_{U X_1 X_2 B} \Vert \rho_{U X_1 X_2 B}^{( \left\{ X_2, B \right\} )}) - 2 - \log \frac{1}{\varepsilon} \\
    R_1 + R_2 &\le D_H^\varepsilon(\rho_{U X_1 X_2 B} \Vert \rho_{U X_1 X_2 B}^{ (\left\{ X_1 X_2, B \right\}) }) - 2 - \log \frac{1}{\varepsilon}.
  \end{align*}
  Then, applying~\cref{lem:packing}, the probability of error in decoding is at most
  \begin{align*}
    p_e & \le f(3, \varepsilon) + 4 \left( 2^{-2 - \log \frac{1}{\varepsilon}} \right) \times 3 \\
    & \le 2^{2^{21}} \varepsilon^{1/3} + 3\varepsilon \le (2^{2^{21}}+3) \varepsilon^{1/3},
  \end{align*}
  where we used the coarse approximation in~\cref{eq:upperBoundF} and that $\varepsilon \in (0,1)$. Using that $X_1$ and $X_2$ are independent conditional on $U$, we obtain up to constants Theorem 2 of~\cite{senInPrep}.
\end{example*}
We expect that~\cref{lem:packing} and~\cref{lem:packingAsympt} can be used in a variety of scenarios to directly generalize results from classical network information theory
, which often hinge on~\cref{lem:cPacking},
to the quantum case. In fact, it is not too difficult to see that an i.i.d.\ variant%
\footnote{This is because we assume i.i.d.\ codewords in~\cref{lem:packingAsympt}, which is sufficient for, e.g., relay, multiple access~\cite{senInPrep}, and broadcast channels~\cite{senInPrep2}.}
of~\cref{lem:cPacking} can be derived from~\cref{lem:packingAsympt}.
More precisely, let $(U,X,Y) \sim p_{UXY}$ be a triple of random variables as in the former.
Consider a DAG $G$ consisting of two vertices, corresponding to random variables $U$ and $X$ with joint distribution $p_{UX}$, and an edge going from the former to the latter.
We set $J=\{1\}$, $\ind(X)=\{1\}$, and $M_1=M$ as the message set.
A visualization of this simple multiplex Bayesian network $(G, (U, X), M, \ind)$ is given in~\cref{fig:relayclassical}.

\begin{figure}
  \centering
  \includegraphics[scale=0.7]{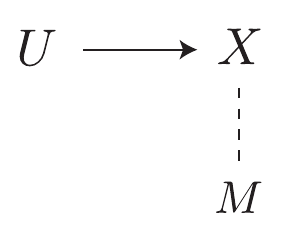}
  \caption{The multiplex Bayesian network $(G, (U, X), M, \ind)$ which relates~\cref{lem:packingAsympt} to~\cref{lem:cPacking}.}\label{fig:relayclassical}
\end{figure}

By running \algref~$n$ times, we obtain codewords which we can identify as $\tilde U^n$ and $\tilde X^n(m)$. Conditioned on $\tilde U^n$, it is clear that for each $m \in M$, $\tilde X^n(m) \sim \bigot_{i=1}^n p_{X \vert U = \tilde U_i}$. Next, choose the subgraph to be all of $G$, set of quantum states to be the classical states
\begin{align*}
  \left\{ \rho_{\tilde Y}^{(u,x)} \equiv \sum_{\tilde y \in \mathcal{Y}}^{} p_{Y \vert U X}(\tilde y \vert u, x) \state{\tilde y}_{\tilde Y} \right\}_{u \in \mathcal{U}, x \in \mathcal{X}},
\end{align*}
and decoding subset $D=\{1\}$, corresponding to $M$.
We see that if we consider the entire system consisting of $\tilde U^n, \, \tilde X^n(m)$ and $\bigot_{i=1}^n \rho_{\tilde Y_i}^{(\tilde U_i \tilde X_i(m'))}$ for $m' \neq m$, it is clear that $\tilde X^n(m)$ is conditionally independent of $\tilde Y^n$ given $\tilde U^n$ due to the conditional independence of $X^n(m)$ and $X^n(m')$ given $\tilde U^n$.
By~\cref{lem:packingAsympt}, we obtain a POVM $\{ Q_{\tilde Y^n}^{(m)} \}_{m \in M}$ such that, for all $m \in M$,
\begin{align*}
  \lim_{n \to \infty} \EE_{C^n} \left[ \tr\left[ \left(I- Q_{\tilde Y^n}^{(m)}\right)\bigot_{i=1}^n \rho_{\tilde Y_i}^{(\tilde u_i x_i(m))}  \right] \right] = 0
\end{align*}
provided $R < I(X;Y \vert U) - \delta(n)$, which is analogous to~\cref{lem:cPacking} if we ``identify'' the POVM measurement with the typicality test.

In~\cref{sec:proofpackinglemma} we prove~\cref{lem:packing} using Sen's quantum joint typicality lemma with $\abs{V}$ classical systems and a single quantum system. We then prove~\cref{lem:packingAsympt}. In the proof of our packing lemma, we actually prove a more general, albeit more abstract, statement.

\section{Application to the Classical-Quantum Relay Channel}\label{sec:quantumrelay}
To illustrate the wide applicability of~\cref{lem:packing} and demonstrate how to use it, we prove a series of achievability results for the classical-quantum relay channel. The first three results make use of the packing lemma in situations where the number of random variables involved in the decoding is at most two ($\abs{V_H} \le 2$). This situation can be dealt with using existing techniques~\cite{savov2012partial}. The final partial decode-forward lower bound, however, applies the packing lemma with $\abs{V_H}$ unbounded with increasing blocklength, thus requiring its full strength. These lower bounds are well-known for classical relay channels~\cite{elgamal2011network}, and our packing lemma allows us to straightforwardly generalize them to the quantum and even finite blocklength case.%
\footnote{Note that in this case the one-shot capacity reduces to the point-to-point scenario, as the relay lags behind the sender.}
We can then invoke~\cref{lem:packingAsympt} to obtain lower bounds on the capacity, which match exactly those of the classical setting with the quantum generalization of mutual information. Note that the partial decode-forward asymptotic bound for the classical-quantum relay channel was first established in~\cite{savov2012partial}.

First we give some definitions. A classical-quantum relay channel~\cite{savov2012partial,jin2012lower} is a classical-quantum channel $\mathcal{N}$ with two classical inputs $X_1, X_2$ and two quantum outputs $B_2, B_3$:
\[ \mathcal{N}_{X_1 X_2 \to B_2 B_3} \colon \mathcal X_1 \times \mathcal X_2 \to \mathcal H_{B_2} \ot \mathcal{H}_{B_3},
\quad (x_1,x_2) \mapsto \rho_{B_2 B_3}^{(x_1 x_2)}. \]
The sender transmits $X_1$, the relay transmits $X_2$ and obtains $B_2$, and the receiver obtains $B_3$. The setup is shown in~\cref{fig:relaysingle}. Note that this is more general than the setting of two concatenated channels because the relay's transmission also affects the system that the relay obtains and the sender's transmission affects the receiver's system.
\begin{figure}[t]
\begin{center}
\includegraphics[scale=0.4]{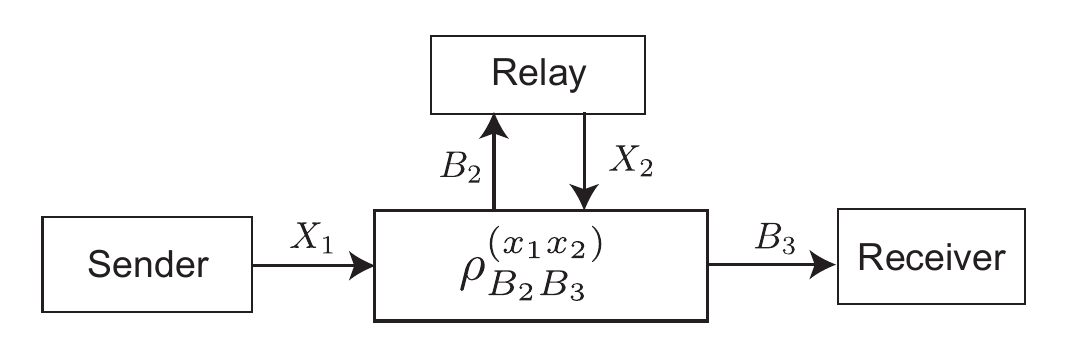}
\end{center}
\caption{The relay channel $\mathcal N_{X_1 X_2 \to B_2 B_3}$. $\rho_{B_2 B_3}^{(x_1 x_2)}$ is a family of quantum states which defines the classical-quantum relay channel.}\label{fig:relaysingle}
\end{figure}

We now define what comprises a general code for the classical-quantum relay channel. Let $n \in \N$, $R \in \R_{\ge 0}$. A $(n, 2^{n R})$ code for classical-quantum relay channel $\mathcal N_{X_1 X_2 \to B_2 B_3}$ for $n$ uses of the channel and number of messages $2^{nR}$ consists of
\begin{enumerate}
  \item A message set $M$ with cardinality $2^{nR}$.
  \item An encoding $x_1^n(m) \in \mathcal{X}_1^n$ for each $m \in M$.
  \item A relay encoding and decoding $\mathcal{R}_{(B_{2})_{j-1} (B_{2}')_{j-1} \to (X_{2})_j (B_{2}')_j}$ for $j \in [n]$. Here, $(B_{2})_j$ is isomorphic to $B_2$ and $(X_2)_j$ isomorphic to $X_2$ while $(B_{2}')_j$ is some arbitrary quantum system. The relay starts with some trivial (dimension 1) quantum system $(B_{2})_0 (B_{2}')_0$.
  \item A receiver decoding POVM $\{Q_{B_3^n}^{(m)}\}_{m \in M}$.
\end{enumerate}
On round $j$, the sender transmits $(x_1)_j(m)$ while the relay applies%
\footnote{The map $\mathcal{R}_{(B_{2})_{j-1} (B_{2}')_{j-1} \to (X_{2})_j (B_{2}')_j}$ depends on~$j$. We do not write~$j$ explicitly since the systems $X_2$, $B_2$ and $B_2'$ are already labeled.}
$\mathcal{R}_{(B_{2})_{j-1} (B_{2}')_{j-1} \to (X_{2})_j (B_{2}')_j}$ to their $(B_{2})_{j-1} (B_{2}')_{j-1}$ system and transmits the $(X_{2})_j$ state while keeping the $(B_{2}')_j$ system. After the completion of $n$ rounds, the receiver applies the decoding POVM $\left\{ Q^{(m)}_{B_3^n} \right\}_{m \in M}$ on their received systems $\rho_{B_3^n}(m)$ to obtain their estimate for the message. See~\cref{fig:relayrepeat} for a visualization of a protocol with $n=3$ rounds.
\begin{figure}[t]
\begin{center}
\includegraphics[scale=0.4]{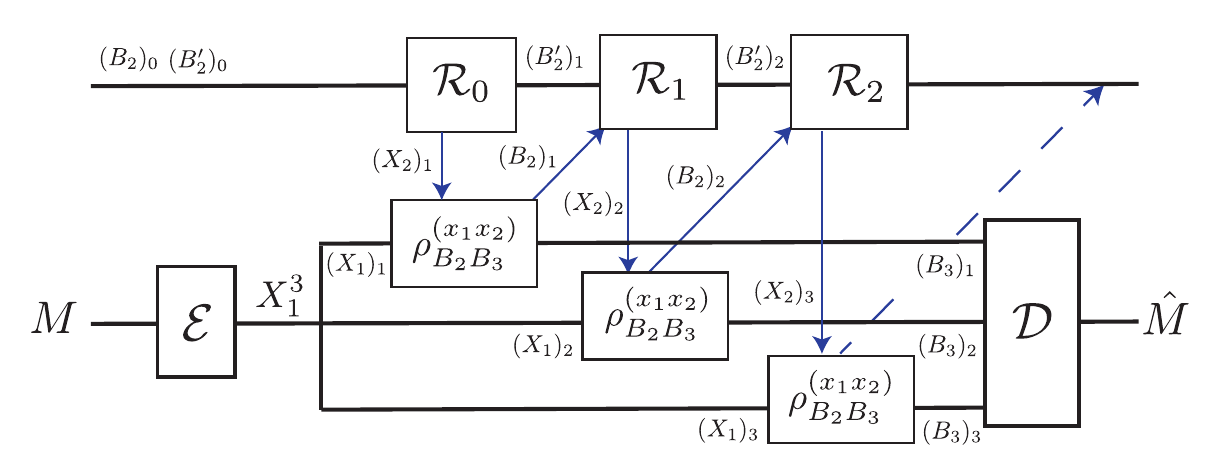}
\end{center}
\caption{A 3-round protocol for the classical-quantum relay channel. Here $\mathcal R_{j-1}$ denotes the relay operation applied on round $j$, and $\mathcal (B'_2)_j$ denotes the state left behind by the relay operation. The decoding operator $\mathcal D$ is applied to all systems $(B_3)_j$ simultaneously.}\label{fig:relayrepeat}
\end{figure}
The average probability of error of a general protocol is given by
\begin{align*}
  p_e =  \frac{1}{|M|} \sum_{m \in M}^{} \tr\left[\left(I- Q_{B_3^n}^{(m)}\right) \rho_{B_3^n}(m)\right].
\end{align*}

In the protocols we give below, we use random codebooks. We can derandomize in the usual way to conform to the above definition of a code. Furthermore, in our protocols the relay only leaves behind a classical system during intermediate stages of the protocol. Since our relay channels are classical-quantum, it is not clear if higher rates can be achieved by letting the relay leave behind a quantum system after every round. We leave the possibility of leaving behind a quantum system in our definition to allow for the most general protocols.

Given $R \in \R_{\ge 0}, \, n \in \N, \, \delta \in [0,1]$, we say that a triple $(R, n, \delta)$ is achievable for a relay channel if there exists a $(n, 2^{n R'})$ code such that
\begin{align*}
  R' \geq R\quad\mathrm{and}\quad p_e  \leq{\delta}.
\end{align*}
The capacity of the classical-quantum relay channel $\mathcal{N}_{X_1 X_2\to B_2 B_3}$ is then defined as
\begin{align*}
  C(\mathcal N) \equiv \lim_{\delta\to0}\liminf_{n\to\infty}\sup\left\{R:\text{$(R, n,\delta)$ is achievable for $\mathcal N$}\right\}.
\end{align*}

Now, before looking at specific coding schemes, we first give a general upper bound, a direct generalization of the cutset bound for the classical relay channel:
\begin{prp}[Cutset Bound]\label{prp:cutset}
  Given a classical-quantum relay channel $\mathcal{N}_{X_1 X_2 \to B_2 B_3}$, its capacity is bounded from above by
  \begin{align}
    \label{eq:cutset}
    & C(\mathcal{N}_{X_1 X_2 \to B_2 B_3}) \nonumber\\
    & \le \max_{p_{X_1 X_2}} \min\left\{ I(X_1 X_2 ; B_3), I(X_1; B_2 B_3 |X_2) \right\}.
  \end{align}
\end{prp}
\begin{proof}
  See~\cref{app:cutset}.
\end{proof}
For some special relay channels,~\cref{prp:cutset} along with some of the lower bounds proven below will be sufficient to determine the capacity.

\subsection{Multihop Scheme}\label{subsec:multihop}
The multihop lower bound is obtained by a simple two-step process where the sender transmits the message to the relay and the relay then transmits it to the receiver. That is, the relay simply ``relays'' the message. The protocol we give below is exactly analogous to the classical case~\cite{elgamal2011network}, right down to the structure of the codebook. In other words, with our packing lemma, the classical protocol can be directly generalized to the quantum case.
The only difference is that the channel outputs a quantum state and the decoding uses a POVM measurement.

Consider a relay channel
\[ \mathcal{N}_{X_1 X_2 \to B_2 B_3} \colon \mathcal X_1 \times \mathcal X_2 \to \mathcal H_{B_2} \ot \mathcal{H}_{B_3},
\quad (x_1,x_2) \mapsto \rho_{B_2 B_3}^{(x_1 x_2)}. \]
Let $R \geq 0$, $b \in \N$, $\eps \in (0,1)$, where $b$ is number of blocks. Again, $R$ is the log of the size of the message set and $b$ the number of relay uses, while $\varepsilon$ is the small parameter input to~\cref{lem:packing}. We show that we can achieve the triple $(\frac{b-1}{b}R, b, \delta)$ for some $\delta$ a function of $R, b, \varepsilon$. Let $p_{X_1}, p_{X_2}$ be probability distributions over $\mathcal{X}_1, \mathcal{X}_2$, respectively. Throughout, we use
\begin{align*}
  & \rho_{X_1 X_2 B_2 B_3} \\
  & \equiv \sum_{x_1,x_2} p_{X_1}(x_1) p_{X_2}(x_2) \ket{x_1x_2}\bra{x_1x_2}_{X_1 X_2} \ot \rho_{B_2 B_3}^{(x_1x_2)}. 
\end{align*}
We also define $\rho_{B_3}^{(x_2)} \equiv \sum_{x_1}^{} p_{X_1}(x_1) \rho_{B_3}^{(x_1 x_2)}$ to be the reduced state on $B_3$ induced by tracing out $X_1 B_2$ and fixing $X_2$.
\\~\\
\noindent \textbf{Code}: Throughout, $j \in [b]$. Let $G$ be a graph with $2b$ vertices corresponding to independent random variables $(X_1)_j \sim p_{X_1},  (X_2)_j \sim p_{X_2}$. Since all the random variables are independent, there are no edges. Furthermore, let $M_0, M_j$ be index sets, where $\abs{M_0} = 1$ and $\abs{M_j} = 2^{R}$. That is, our index set for the different message sets should be $J = [0:b]$. The $M_j$ are the sets from which the messages for each round is taken. We use a singleton $M_0$ to make the effect of the first and the last blocks more explicit. Finally, the function $\ind$ maps $(X_1)_j$ to $\{j\}$ and $(X_2)_j$ to $\{j-1\}$. Then, letting $X \equiv X_1^b X_2^b$ and $M \equiv \bigtimes_{j=0}^b M_j$, $\mathcal{B} \equiv (G, X, M, \ind)$ is a multiplex Bayesian network. See~\cref{fig:multihop} for a visualization when $b=3$.
\begin{figure}[h]
\begin{center}
\includegraphics[scale=0.45]{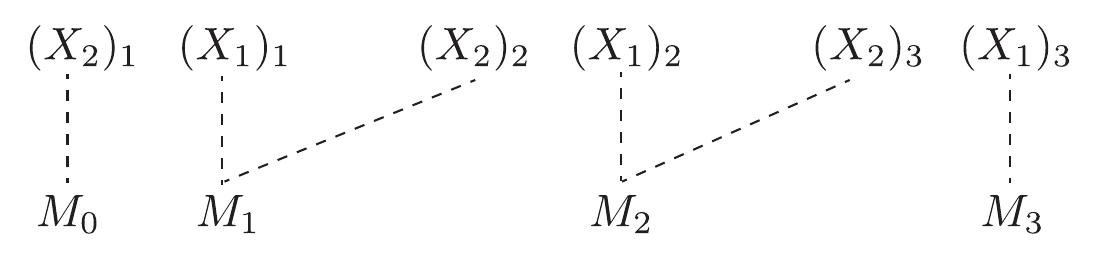}
\end{center}
\caption{Bayesian multiplex network $\mathcal{B}$ that generates the codebook $C$ for the multihop scheme with $b=3$. The absence of solid edges indicates that each $(X_i)_j$ is sampled independently.}\label{fig:multihop}
\end{figure}
Now, run \algref with $\mathcal{B}$ as the argument. This returns a random codebook
\begin{align*}
  C = \bigcup_{j=1}^b \left\{ (x_1)_j(m_j), (x_2)_j(m_{j-1})\right\}_{ m_j \in M_j, m_{j-1} \in M_{j-1}},
\end{align*}
where we restricted to the message indices the codewords are dependent on via $\ind$. For decoding we apply~\cref{lem:packing} with this codebook and use the assortment of POVMs that are given for different ancestral subgraphs and other parameters.
\\~\\
\noindent \textbf{Encoding}: On the $j$\textsuperscript{th} transmission, the sender transmits a message $m_j \in M_j$ via $(x_1)_j(m_j) \in C$.
\\~\\
\noindent \textbf{Relay encoding}: Set $\tilde m_0$ to be the sole element of $M_0$. On the $j$\textsuperscript{th} transmission, the relay sends their estimate $\tilde m_{j-1}$ via $(x_2)_j(\tilde m_{j-1}) \in C$. Note that this is the relay's estimate of the message $m_{j-1}$ transmitted by the sender on the $(j-1)$\textsuperscript{th} transmission.
\\~\\
\noindent \textbf{Relay decoding}: Consider the $j$\textsuperscript{th} transmission. We invoke~\cref{lem:packing} with the ancestral subgraph containing the two vertices $(X_1)_j$ and $(X_2)_j$, the set of quantum states $\left\{ \rho_{B_2}^{(x_1 x_2 )} \right\}_{x_1 \in \mathcal{X}_1, x_2 \in \mathcal{X}_2}$, decoding subset $\left\{  j\right\} \subseteq \left\{ j-1, j \right\}$, and small parameter $\eps \in (0,1)$. The relay picks the POVM corresponding to the message estimate for the previous round $\tilde m_{j-1}$, which is denoted by $\left\{ Q_{B_2}^{(m_j' | \tilde m_{j-1})} \right\}_{m_j' \in M_j}$. The relay applies this on their received state to obtain a measurement result $\tilde m_j$. Note that this is the relay's estimate for message $m_j$.
\\~\\
\noindent \textbf{Decoding}: 
On the $j$\textsuperscript{th} transmission, we again invoke~\cref{lem:packing} and let the receiver use the POVM corresponding to the ancestral subgraph containing just the vertex $(X_2)_{j}$, the set of quantum states $\left\{ \rho_{B_3}^{(x_2)} \right\}_{x_2 \in \mathcal{X}_2}$, decoding subset $\{ j-1 \} \subseteq \left\{ j-1 \right\}$, and small parameter $\eps$. Note that we don't have a message guess here since the decoding subset is not a proper subset. In this case we suppress the conditioning for conciseness. We denote the POVM by $\left\{ Q_{B_3}^{(m_{j-1}')} \right\}_{m_{j-1}' \in M_{j-1}}$, and the receiver applies this on their received state to obtain a measurement result $\hat m_{j-1}$. Note that this is the receiver's estimate of the $(j-1)$\textsuperscript{th} message. $\hat m_0$ is trivially be the sole element of $M_0$.
\\~\\
\noindent \textbf{Error analysis}:
Set $m_0$ to be the sole element of $M_0$. Fix ${\bf m} \equiv (m_0, \dots, m_{b-1})$. Note that $m_b$ is never decoded by the receiver since it is the message sent in the last block and thus, we can ignore it without loss of generality. Let $\tilde {\bf m} \equiv (\tilde m_0, \dots, \tilde m_{b-1}), \hat {\bf m} \equiv (\hat m_0, \dots, \hat m_{b-1})$ denote the aggregation of the message estimates of the relay and receiver, respectively.%
\footnote{Note that $m_0 = \tilde m_0 = \hat m_0$ and refer to the same sole element of $M_0$.}
The probability of error averaged over the random codebook $C$ is given by
\begin{align*}
  p_e(C) = \EE_C \left[ p(\hat {\bf m} \neq {\bf m}) \right],
\end{align*}
where $p$ here denotes the probability for a fixed codebook. Now, by~\cref{eq:easybound},
\begin{align}
  \label{eq:multihopUnionBound}
  p_e(C) \le \EE_C \left[ p(\tilde {\bf m} \neq {\bf m}) \right] + \EE_C \left[ p(\hat {\bf m} \neq {\bf m} | \tilde {\bf m} = {\bf m}) \right].
\end{align}
We consider the first term corresponding to the relay decoding. By the union bound,
\begin{align*}
  \EE_C \left[ p(\tilde {\bf m} \neq {\bf m})  \right] \le \, & \EE_C \left[ p(\tilde m_0 \neq m_0) \right] \\
  & + \sum_{j=1}^{b-1} \EE_C \left[ p(\tilde m_j \neq m_j | \tilde m_{j-1} = m_{j-1}) \right].
\end{align*}
By the definition of $\tilde m_0$, the first term is zero. Now, we can apply~\cref{eq:packing} to bound each summand in the second term as follows:%
\footnote{The careful reader would notice that the conditioning on $\tilde m_{j-1} = m_{j-1}$ is not necessary here since the probability of decoding $m_j$ correctly at the relay is independent of whether $m_{j-1}$ was decoded successfully. However, this will be necessary for the other schemes we give.}
\begin{align*}
  & \EE_C \left[ p(\tilde m_j \neq m_j | \tilde m_{j-1} = m_{j-1}) \right] \\
  & = \EE_{C} [ \tr[ (I - Q^{(m_j |m_{j-1})}_{B_2}) \rho^{((x_1)_j(m_j) (x_2)_j(m_{j-1}))}_{B_2} ] ] \\
  & = \EE_{C_{(X_1)_j (X_2)_j}} [ \tr[ (I - Q^{(m_j |m_{j-1})}_{B_2}) \rho^{((x_1)_j(m_j) (x_2)_j(m_{j-1}))}_{B_2} ] ] \\
  & \le f(2, \varepsilon) +  4 \sum_{T = \left\{ j \right\} }^{} 2^{R -D_H^\varepsilon(\rho_{(X_1)_j (X_2)_j B_2} \Vert \rho_{(X_1)_j (X_2)_j B_2}^{(\{X_{S_T}, B_2 \})})}, \\
  & = f(2, \varepsilon) + 4 \times 2^{R -D_H^\eps(\rho_{X_1 X_2 B_2} \Vert \rho_{X_1 X_2 B_2}^{(\left\{  X_1, B_2\right\})})},
\end{align*}
where $C_{(X_1)_j (X_2)_j}$ is the corresponding subset of the codebook $C$, and we used $S_{ \left\{ j \right\}} = \left\{ (X_1)_j \right\}$. We dropped the index $j$ in the last equality since $(X_1)_j (X_2)_j \sim p_{X_1} \times p_{X_2}$. Hence, overall,
\begin{align*}
  & \EE_C \left[ p(\tilde {\bf m} \neq {\bf m}) \right] \\
  & \le b \left[  f(2,\varepsilon) + 4 \times 2^{R-D_H^\eps(\rho_{X_1 X_2 B_2} \Vert \rho_{X_1 X_2 B_2}^{(\left\{  X_1, B_2 \right\})})}\right].
\end{align*}

We now consider the second term in~\cref{eq:multihopUnionBound}, corresponding to the receiver decoding. By the union bound,
\begin{align*}
  \EE_C \left[ p(\hat {\bf m} \neq {\bf m} | \tilde {\bf m} = {\bf m}) \right] \le  & \, \EE_C \left[ p(\hat m_0 \neq m_0 | \tilde{{\bf m}} = {\bf m}) \right] \\
  & +  \sum_{j=1}^{b-1} \EE_C \left[p(\hat m_j \neq m_j | \tilde {\bf m} = {\bf m} )\right].
\end{align*}
Again by definition, the first term vanishes. Now, the receiver on the $(j+1)$\textsuperscript{th} transmission obtains the state $\rho_{B_3}^{((x_1)_{j+1}(m_{j+1}) (x_2)_{j+1}(\tilde m_{j}))}$. Averaging over $(x_1)_{j+1}(m_{j+1})$, this becomes $\rho_{B_3}^{((x_2)_{j+1}(\tilde m_{j}))}$. Hence, the summands in second term are also bounded via~\cref{eq:packing}:
\newpage
\begin{align*}
  & \EE_C \left[ p(\hat m_j \neq m_j | \tilde {\bf m}= {\bf m}) \right] \\
  & = \EE_{C} \left[ \tr\left[  (I - Q^{(m_j)}_{B_3}) \rho^{((x_1)_{j+1}(m_{j+1}) (x_2)_{j+1}( m_j))}_{B_3}\right] \right] \\
  & = \EE_{C_{(X_1)_{j+1} (X_2)_{j+1}}} \Big[ \tr \Big[  (I - Q^{(m_j)}_{B_3}) \\
  & \quad \, \, \rho^{((x_1)_{j+1}(m_{j+1}) (x_2)_{j+1}( m_j))}_{B_3}\Big] \Big] \\
  & = \EE_{C_{(X_2)_{j+1}}} \left[ \tr\left[  (I - Q^{(m_j)}_{B_3}) \rho^{( (x_2)_{j+1}(m_j))}_{B_3}\right] \right] \\
  & \le f(1, \varepsilon) + 4 \sum_{T = \left\{ j \right\}}^{} 2^{R -D_H^\varepsilon(\rho_{(X_2)_{j+1} B_3} \Vert \rho_{(X_2)_{j+1} B_3}^{(\{X_{S_T}, B_3 \})})}, \\
  & \le f(1, \varepsilon) + 4 \times 2^{R -D_H^\eps(\rho_{X_2 B_3} \Vert \rho_{X_2 B_3}^{(\left\{ X_2, B_3 \right\})})},
\end{align*}
where we used $S_{ \left\{ j \right\}} = \left\{ (X_2)_{j+1} \right\}$ and again dropped indices in the last inequality. Hence, overall
\begin{align*}
  & \EE_C \left[ p(\hat {\bf m} \neq {\bf m} | \tilde {\bf m} = {\bf m}) \right] \\
  & \le b\left[  f(1, \varepsilon) + 4 \times 2^{R -D_H^\eps(\rho_{X_2 B_3} \Vert \rho_{X_2 B_3}^{(\left\{ X_2, B_3 \right\})})}\right].
\end{align*}
Note that since $X_1, X_2$ are independent, $\rho_{X_2 B_3}^{(\{X_2, B_3\})} = \rho_{X_2} \ot \rho_{B_3}$. We have therefore established the following:
\begin{prp}[Multihop]
  Given $R \in \R_{\ge 0}, \, \varepsilon \in (0,1), \, b \in \N$, the triple $(\frac{b-1}{b} R,b, \delta)$, is achievable for the classical-quantum relay channel, where%
\footnote{Note that we need $R, b, \varepsilon$ to be sufficiently small so that $\delta \in [0,1]$. Otherwise, a block Markov scheme can be employed to obtain a meaningful error bound.}
\begin{align*}
  \delta=  b\big[ & f(1, \varepsilon) + f(2, \varepsilon) +  4 \times 2^{R -D_H^\eps(\rho_{X_2 B_3} \Vert \rho_{X_2 B_3}^{(\left\{ X_2, B_3 \right\})})} \\
  & + 4 \times 2^{R-D_H^\eps(\rho_{X_1 X_2 B_2} \Vert \rho_{X_1 X_2 B_2}^{(\left\{  X_1, B_2 \right\})})}\big].
\end{align*}
\end{prp}
In the asymptotic limit we use the channel $n/b$ times in each of the $b$ blocks. The protocol is analogous to one-shot protocol, except the relay channel has a tensor product form $ \mathcal{N}^{\otimes (n/b)}_{X_1 X_2 \to B_2 B_3}$ characterized by a family of quantum states $\rho^{(x_1^{(n/b)} x_2^{(n/b)})}_{B_2^{(n/b)} B_3^{(n/b)}}$. The codebook is $C^{(n/b)}$ and for finite $b$ and large $n$ we invoke~\cref{lem:packingAsympt} (instead of~\cref{lem:packing}) to construct POVM's for the relay and the receiver such that the decoding error vanishes if the rate satisfies $R <\min\left\{ I(X_1 ; B_2 | X_2)_\rho , I(X_2; B_3)_\rho \right\}$, thereby obtaining the quantum equivalent of the classical multihop bound for sufficiently large $n,b$:%
\footnote{Note that our rate is $\frac{b-1}{b} R$. To achieve rate $R$ we need $\frac{b-1}{b} \to 1$, and so we take the large $n$ limit \emph{followed by} the large $b$ limit.}
\begin{align}
  C \ge \max_{p_{X_1} p_{X_2}} \min\left\{ I(X_1 ; B_2 | X_2)_\rho , I(X_2; B_3)_\rho \right\}. \label{eq:multihop-rate}
\end{align}

\subsection{Coherent Multihop Scheme}\label{subsec:coherentmultihop}
In the multihop scheme, we obtained a rate optimized over product distributions, specifically~\cref{eq:multihop-rate}. For the coherent multihop scheme we obtain the same rate except optimized over all possible two-variable distributions $p_{X_1 X_2}$ by conditioning codewords on each other.

Again, let $R \geq 0$ be our rate, $\eps \in (0,1)$, and total blocklength $b \in \N$. We show that we can achieve the triple $(\frac{b-1}{b} R, b, \delta)$ for some $\delta$ a function of $R, b, \varepsilon$. Let $p_{X_1 X_2}$ be probability distributions over $\mathcal{X}_1 \times \mathcal{X}_2$. Throughout, we use
\begin{align*}
 & \rho_{X_1 X_2 B_2 B_3} \\
 & \equiv \sum_{x_1,x_2} p_{X_1 X_2}(x_1, x_2)  \ket{x_1x_2}\bra{x_1x_2}_{X_1 X_2} \ot \rho_{B_2 B_3}^{(x_1x_2)}. 
\end{align*}
We also again define $\rho_{B_3}^{(x_2)} \equiv \sum_{x_1}^{} p_{X_1|X_2}(x_1\vert x_2) \rho_{B_3}^{(x_1 x_2)}$ to be the reduced state on $B_3$ by tracing out $X_1 B_2$ and fixing $X_2$. Our coding scheme is similar to that of the multihop.
\\~\\
\noindent \textbf{Code}: Let $G$ be a graph with $2b$ vertices corresponding to random variables $(X_1)_j (X_2)_j \sim p_{X_1 X_2}$, independent of other pairs, with edges from $(X_1)_j$ to $(X_2)_j$. Furthermore, let $M_0, M_j$ be index sets, where $\abs{M_0} = 1$ and $\abs{M_j} = 2^{R}$. Finally, the function $\ind$ maps $(X_1)_j$ to $\{j\}$ and $(X_2)_j$ to $\{j-1\}$. Then, letting $X \equiv X_1^b X_2^b$ and $M \equiv \bigtimes_{j=0}^b M_j$, it is easy to see that $\mathcal{B} \equiv (G, X, M, \ind)$ is a multiplex Bayesian network. See~\cref{fig:coherentmultihop} for a visualization when $b=3$.
\begin{figure}[t]
\begin{center}
  \includegraphics[scale=0.45]{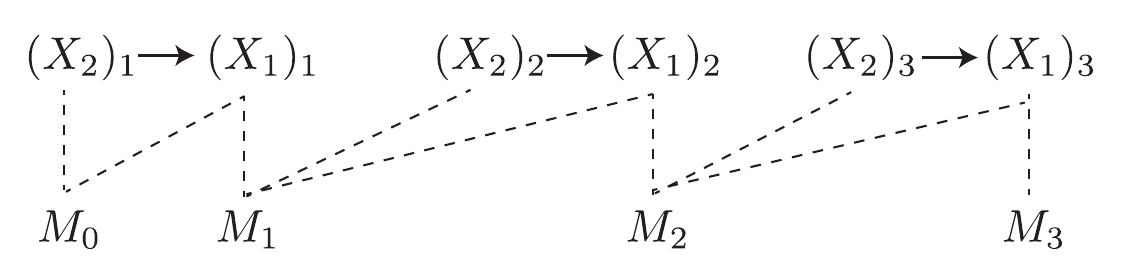}
\end{center}
\caption{Bayesian multiplex network $\mathcal{B}$ that generates the codebook $C$ for the coherent multihop scheme with $b=3$ blocks. }
\label{fig:coherentmultihop}
\end{figure}
Now, run \algref with $\mathcal{B}$ as the argument to obtain a random codebook $C$ given by
\begin{align*}
  \bigcup_{j=1}^b \left\{ (x_1)_j(m_{j-1}, m_{j}), (x_2)_j(m_{j-1})\right\}_{ m_j \in M_j, m_{j-1} \in M_{j-1}},
\end{align*}
where we restricted to the message indices the codewords depend on via $\ind$. For decoding we apply~\cref{lem:packing} with this codebook and use the assortment of POVMs that are given for different ancestral subgraphs and other parameters. \\
\\
\noindent \textbf{Encoding}: Set $m_0$ to be the sole element of $M_0$. On the $j$\textsuperscript{th} transmission, the sender transmits a message $m_j \in M_j$ via $(x_1)_j(m_{j-1},m_j) \in C$.
\\~\\
\noindent \textbf{Relay encoding}: Same as multihop.
\\~\\
\noindent \textbf{Relay decoding}: Same as multihop.%
\footnote{Note, however, that the POVM the relay uses from~\cref{lem:packing} is \emph{not} be the same as that of the multihop case since the multiplex Bayesian networks are not the same.}
\\~\\
\noindent \textbf{Decoding}: Same as multihop.
\\~\\
\noindent \textbf{Error analysis}:
With an analysis essentially identical to that of the multihop protocol we arrive at the following.
\begin{prp}[Coherent Multihop]
  Given $R \in \R_{\ge 0}, \, \varepsilon \in (0,1), \, b \in \N$, the triple $(\frac{b-1}{b} R,b, \delta)$ is achievable for the classical-quantum relay channel, where
\begin{align*}
  \delta = b\big[ & f(1,\varepsilon) + f(2,\varepsilon) +  4 \times 2^{R -D_H^\eps(\rho_{X_2 B_3} \Vert \rho_{X_2 B_3}^{\left\{ X_2, B_3 \right\}})} \\
  & + 4 \times 2^{R-D_H^\eps(\rho_{X_1 X_2 B_2} \Vert \rho_{X_1 X_2 B_2}^{(\left\{  X_1, B_2 \right\})})}\big].
\end{align*}
\end{prp}
Asymptotically, this vanishes if 
\[ R< \min\left\{ I(X_1 ; B_2 | X_2)_\rho , I(X_2; B_3)_\rho \right\},\]
thereby obtaining the quantum equivalent of the coherent multihop bound for sufficiently large $b$:
\begin{align*}
  C \ge \max_{p_{X_1 X_2}} \min\left\{ I(X_1 ; B_2 | X_2)_\rho , I(X_2; B_3)_\rho \right\}.
\end{align*}

\subsection{Decode-Forward Scheme}\label{subsec:decodeforward}
In the decode-forward protocol we make an incremental improvement on the coherent multihop protocol by letting the receiver's decoding also involve $X_1$.

Again, let $R \geq 0$ be our rate, $\eps \in (0,1)$, and total number of blocks $b \in \N$. The classical-quantum state $\rho_{X_1 X_2 B_2 B_3}$ is identical to that of the coherent multihop scenario.
\\~\\
\noindent \textbf{Code}: The codebook is generated in the same way as in the coherent multihop protocol save with the index set $M_b$ having cardinality 1 to take into account boundary effects for the backward decoding protocol%
\footnote{In~\cite{elgamal2011network} multiple decoding protocols are given. We here give the quantum generalization of the backward decoding protocol.}
we implement.
\\~\\
\noindent \textbf{Encoding}: Set $m_0$ to be the sole element of $M_0$. On the $j$\textsuperscript{th} transmission, the sender transmits the message $m_j \in M_j$ via $(x_1)_j (m_{j-1}, m_j ) \in C $. Note that there is only one message $m_b \in M_b$ they can choose on the $b$\textsuperscript{th} round.
\\~\\
\noindent \textbf{Relay encoding}: Same as that of coherent multihop.
\\~\\
\noindent \textbf{Relay decoding}: Same as that of coherent multihop. However, note that on $b$\textsuperscript{th} round, since $\abs{M_b}=1$, the decoding is trivial and the estimate $\tilde m_b$ is the sole element of $M_b$.
\\~\\
\noindent \textbf{Decoding}: The receiver waits until all $b$ transmissions are finished. Then, they implement a backward decoding protocol, that is, starting with the last system they obtain. Set $\hat m_b$ to be the sole element of $M_b$. On the $j$\textsuperscript{th} system they use the POVM corresponding to the ancestral subgraph containing vertices $(X_1)_{j}$ and $(X_2)_{j}$, the set of quantum states $\left\{ \rho_{B_3}^{(x_1 x_2)} \right\}_{x_1 \in \mathcal{X}_1, x_2 \in \mathcal{X}_2}$, decoding subset $\{ j-1 \} \subseteq \left\{ j-1, j \right\}$, and small parameter $\eps$. We denote the POVM by $\left\{ Q_{B_3}^{(m_{j-1}' | \hat m_{j})} \right\}_{m_{j-1}' \in M_{j-1}}$, where we use the estimate $\hat m_{j}$, and the obtained measurement result $\hat m_{j-1}$. Note that trivially $\hat m_0$ is the sole element of $M_0$.
\\~\\
\noindent \textbf{Error analysis}: Fix some ${\bf m} = (m_0, \dots, m_{b}) \in M$. Let $\tilde {\bf m} = (\tilde m_0, \dots, \tilde m_{b}), \hat {\bf m} = (\hat m_0, \dots, \hat m_{b})$ denote the aggregation of the messages estimates of the relay and receiver, respectively. Then, the probability of error averaged over $C$ is given by
\begin{align*}
  p_e(C) = \EE_C \left[ p(\hat {\bf m} \neq {\bf m}) \right].
\end{align*}
Again, by the bound in~\cref{eq:easybound},
\begin{align*}
  p_e(C) \le \EE_C \left[ p(\tilde {\bf m} \neq {\bf m}) \right] + \EE_C \left[ p(\hat {\bf m} \neq {\bf m} | \tilde {\bf m} = {\bf m}) \right].
\end{align*}
The bound on the first term is identical to that of the coherent multihop protocol and is given by
\begin{align*}
  & \EE_C \left[ p(\tilde {\bf m} \neq {\bf m}) \right] \\
  & \le b \left[  f(2,\varepsilon) + 4 \times 2^{R-D_H^\eps(\rho_{X_1 X_2 B_2} \Vert \rho_{X_1 X_2 B_2}^{(\left\{  X_1, B_2 \right\})})}\right].
\end{align*}
For the second term, we first apply the union bound:
\begin{align*}
  & \EE_C \left[ p(\hat {\bf m} \neq {\bf m} | \tilde {\bf m} = {\bf m}) \right] \\
  & \le  \EE_C \left[ \sum_{j=1}^{b-1} p(\hat m_j \neq m_j | \hat m_{j+1} = m_{j+1}\land \tilde {\bf m}= {\bf m} )\right],
\end{align*}
where we take into account that the terms corresponding to $0$ and $b$ vanish by definition. Each of the summands can be bounded via~\cref{lem:packing}:
\begin{align*}
  & \EE_C \left[ p(\hat m_j \neq m_j | \hat m_{j+1} = m_{j+1} \land \tilde {\bf m}= {\bf m}) \right] \\
  & = \EE_{C} \left[ \tr\left[  (I - Q^{(m_j | m_{j+1})}_{B_3}) \rho^{((x_1)_{j+1}(m_{j+1} |m_j) (x_2)_{j+1}(m_j))}_{B_3}\right] \right] \\
  & = \EE_{C_{(X_1)_{j+1} (X_2)_{j+1}}} \Big[ \tr\Big[  (I - Q^{(m_j | m_{j+1})}_{B_3}) \\
  & \quad \, \, \rho^{((x_1)_{j+1}(m_{j+1} |m_j) (x_2)_{j+1}(m_j))}_{B_3}\Big] \Big] \\
  & \le f(2,\varepsilon)  \\
  & \quad + 4 \sum_{T =\{j \}} 2^{R -D_H^\eps(\rho_{(X_1)_{j+1} (X_2)_{j+1} B_3} \Vert \rho_{(X_1)_{j+1} (X_2)_{j+1} B_3}^{(\{X_{S_T}, B_3\})})} \\
  & \le f(2,\varepsilon) + 4 \times 2^{R -D_H^\eps(\rho_{X_1 X_2 B_3} \Vert \rho_{X_1 X_2 B_3}^{(\left\{ X_1 X_2, B_3 \right\})})},
\end{align*}
where we use that $S_{\{ j \}}=\{ (X_1)_{j+1} (X_2)_{j+1}\}$. Hence, we conclude that
\begin{align*}
  & \EE_C \left[ p(\hat {\bf m} \neq {\bf m} | \tilde {\bf m} = {\bf m}) \right] \\
  & \le b\left[  f(2,\varepsilon) + 4 \times 2^{R -D_H^\eps(\rho_{X_1 X_2 B_3} \Vert \rho_{X_1 X_2 B_3}^{(\left\{X_1 X_2, B_3 \right\})})}\right].
\end{align*}
We conclude the following.
\begin{prp}[Decode-Forward]
  Given $R \in \R_{\ge 0}, \, \varepsilon \in (0,1), \, b \in \N$, the triple $(\frac{b-1}{b} R, b , \delta)$ is achievable for the classical-quantum relay channel where
\begin{align*}
  \delta =  b\Big[ 2f(2,\varepsilon) +  4 \times \Big(& 2^{R -D_H^\eps(\rho_{X_1 X_2 B_3} \Vert \rho_{X_1 X_2 B_3}^{(\left\{ X_1 X_2, B_3 \right\})})} \\
  & + 2^{R-D_H^\eps(\rho_{X_1 X_2 B_2} \Vert \rho_{X_1 X_2 B_2}^{(\left\{  X_1, B_2 \right\})})} \Big) \Big].
\end{align*}
\end{prp}
Asymptotically, this vanishes if 
\[ R  < \min\left\{ I(X_1 ; B_2 | X_2)_\rho , I(X_1 X_2; B_3)_\rho \right\},\]
thereby obtaining the decode-forward lower bound for sufficiently large $b$:
\begin{align*}
  C \ge \max_{p_{X_1 X_2}} \min\left\{ I(X_1 ; B_2 | X_2)_\rho , I(X_1 X_2; B_3)_\rho \right\}.
\end{align*}

\subsection{Partial Decode-Forward Scheme}\label{partialdecodeforward}
We now derive the partial decode-forward lower bound. This requires the full power of~\cref{lem:packing} as the receiver decodes all the messages simultaneously by performing a joint measurement on all $b$ blocks. Intuitively, the partial decode-forward builds on the decode-forward by letting the relay only decode and pass on a part, what we call $P$, of the overall message.

We split the message into two parts $P$ and $Q$ with respective rates $R_p, R_q \geq 0$. Let $\eps \in (0, 1)$ and $b \in \N$ be the total blocklength. Choose some distribution $p_{X_1 X_2}$ but also a random variable $U$ correlated with $X_1 X_2$ so that the overall distribution is $p_{U X_1 X_2}$. The classical-quantum state of interest is
\begin{align}
  \label{eq:cqStatePartialDecode}
  \rho_{U X_1 X_2 B_2 B_3} \equiv \sum_{u, x_1,x_2} & p_{U X_1 X_2 } ( u, x_1, x_2)  \nonumber\\
  & \ket{u x_1x_2 }\bra{ u x_1x_2 }_{U X_1 X_2 } \ot \rho_{B_2 B_3}^{(x_1x_2 )}.
\end{align}
Note that $\rho_{B_2 B_3}^{(x_1 x_2)}$ does \emph{not} depend on $u$, but sometimes we will write $\rho_{B_2 B_3}^{(u x_1 x_2)} (= \rho_{B_2 B_3}^{(x_1 x_2)})$ to keep notation explicit. However, if we trace over $X_1$, we induce a $u$ dependence via the correlation between $U$ and $X_1 X_2$:
\begin{align*}
  \rho_{U X_2 B_2 B_3} = \sum_{u, x_2} p_{U X_2 } ( u, x_2)  \ket{u x_2 }\bra{ u x_2 }_{U X_2 } \ot \rho_{B_2 B_3}^{(u x_2 )},
\end{align*}
where
\begin{align*}
  \rho_{B_2 B_3}^{(u x_2)} \equiv \sum_{x_1}^{} p_{X_1 | U X_2} (x_1 | u, x_2) \rho_{B_2 B_3}^{(x_1 x_2)}.
\end{align*}
This state will be important for the relay decoding.
\\~\\
\noindent \textbf{Code}: Let $G$ be a graph with $3b$ vertices corresponding to random variables $(U)_j (X_1)_j (X_2)_j \sim p_{U X_1 X_2}$. The graph has edges going from $(X_2)_j$ to $(U)_j$ and $(U)_j$ to $(X_1)_j$ for all $j$ and no edges going across blocks with different $j$'s. Furthermore, let $P_0, P_j$ and $Q_j$ be index sets, so that $J = [0:b] \sqcup [b]$, where $\abs{P_0} = \abs{P_b}=\abs{Q_b} = 1$, $\abs{P_j} = 2^{R_p}$ and $\abs{Q_j} = 2^{R_q}$ otherwise. Finally, the function $\ind$ maps $(X_1)_j$ to%
\footnote{For convenience we denote the elements of $J$ by the index sets they correspond to.}
$\{P_j, Q_j, P_{j-1}\}$, $(U)_j$ to $\{P_j, P_{j-1}\}$, and $(X_2)_j$ to $\{P_{j-1}\}$. Then, letting $X \equiv U^b X_1^b X_2^b$, $M_p = \bigtimes_{j=0}^b P_j$, $M_q = \bigtimes_{j=1}^b Q_j$ and $M = M_p \times M_q$, it is easy to see that $\mathcal{B} \equiv (G, X, M, \ind)$ is a multiplex Bayesian network. See~\cref{fig:partialdecode} for a visualization when $b=3$.
\begin{figure}[t]
\begin{center}
\includegraphics[scale=0.325]{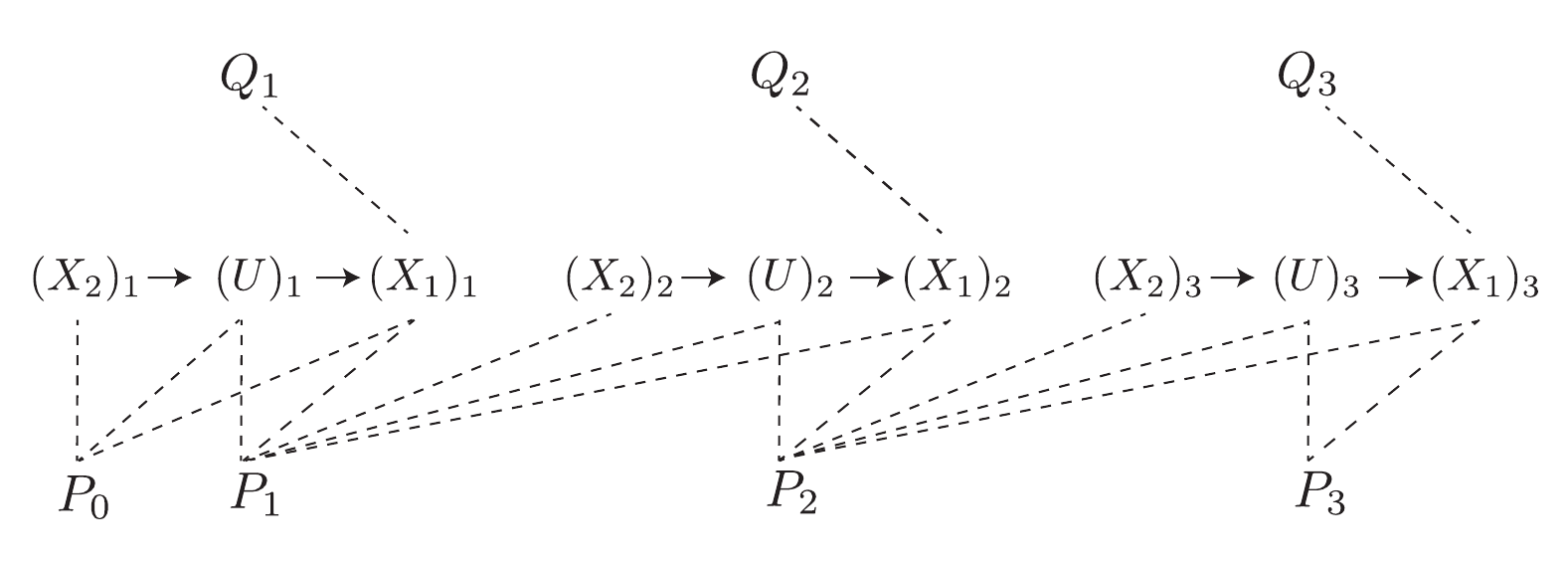}
\end{center}
\caption{Bayesian multiplex network $\mathcal{B}$ that generates the codebook $C$ for the partial decode-forward scheme with $b=3$ blocks.}\label{fig:partialdecode}
\end{figure}
Now, run \algref with $\mathcal{B}$ as the argument. This returns a random codebook
\begin{align*}
  C = \bigcup_{j=1}^b \{ & (x_1)_j(p_{j-1}, p_j, q_j ), (u)_j(p_{j-1}, p_j ), \\
  & (x_2)_j(p_{j-1})\}_{ p_j \in P_j, p_{j-1} \in P_{j-1}, q_j \in Q_{j}},
\end{align*}
where we restricted to the message indices the codewords are dependent on via $\ind$. For decoding we apply~\cref{lem:packing} with this codebook and use the assortment of POVMs that are given for different ancestral subgraphs and other parameters.
\\~\\
\noindent \textbf{Encoding}: Set $p_0$ to be the sole element of $P_0$. On the $j$\textsuperscript{th} transmission, the sender transmits the two-part message $(p_j,q_j) \in P_j \times Q_j$ via $(x_1)_j(p_{j-1}, p_j, q_j ) \in C$. Note that on the $b$\textsuperscript{th} transmission the sender has to send a fixed message $(p_b, q_b)$ being the sole element of $P_b \times Q_b$.
\\~\\
\noindent \textbf{Relay encoding}: Let $\tilde p_0$ to be the sole element of $P_0$. On the $j$\textsuperscript{th} transmission, the relay sends $\tilde p_{j-1}$ via $(x_2)_j (\tilde p_{j-1})$ from codebook $C$. Note that this is the relay's estimate of the message sent by the sender on the $(j-1)$\textsuperscript{th} transmission.
\\~\\
\noindent \textbf{Relay decoding}: The relay tries to recover the $p$-part of the sender's message using the same technique as in the previous protocols. On the $j$\textsuperscript{th} transmission the relay uses the POVM corresponding to the ancestral subgraph containing the two vertices $(U)_j$ and $(X_2)_j$, the set of quantum states $\left\{ \rho_{B_2}^{(u x_2 )} \right\}_{u \in \mathcal{U}, x_2 \in \mathcal{X}_2}$, decoding subset $\left\{  P_j\right\} \subseteq \left\{ P_{j-1}, P_j \right\}$, and small parameter $\eps$. The POVM is denoted by $\Big\{ Q_{B_2}^{(p_j' | \tilde p_{j-1})} \Big\}_{p_j' \in P_j}$, where we use the estimate $\tilde p_{j-1}$, and the relay applies this on their received state to obtain a measurement result $\tilde p_j$. Note that $\tilde p_b$ is trivially the sole element of $P_b$.
\\~\\
\noindent \textbf{Decoding}: The decoder waits until all $b$ transmissions are completed.
The receiver uses the POVM corresponding to the ancestral subgraph the entire graph $G$, the set of quantum states $\left\{ \bigotimes_{j=1}^b \rho_{B_3}^{((u)_j (x_1)_j (x_2)_j)} \right\}$, where the $(u)_j$ dependence here is trivial, decoding set $\bigtimes_{j=1}^{b-1} P_j \times \bigtimes_{j=1}^{b-1} Q_j$, and small parameter $\eps$. We denote the POVM by%
\footnote{Since the only index sets which are not included in the part to be decoded are all of cardinality 1, we omit the conditioning for conciseness.}
$\left\{Q^{(p'_1 p'_2 \cdots p'_{b-1},    q_1' q_2' \cdots q_{b-1}' )}_{B^b_3} \right\}_{ \bigtimes_{j=1}^{b-1}  p_j' \in P_j, q_j' \in Q_j}$, to their received state on $B_3^b$ to obtain their estimate of the entire string of messages, which we call $\hat m_p \equiv (\hat p_0, \dots, \hat p_b), \hat m_q \equiv (\hat q_1, \dots, \hat q_b) $, where $\hat p_0, \hat p_b, \hat q_b$ are set to be the sole elements of the respective index sets.
\\~\\
\noindent \textbf{Error analysis}:
We fix the strings of messages $m_p = (p_0, \dots, p_b)$ and $m_q = (q_1, \dots, q_b)$. By the bound in~\cref{eq:easybound},
\newpage
\begin{align*}
  p_e(C) & \equiv \EE_C[ p(\hat m_p \hat m_q \neq m_p m_q)] \\
  & \le \EE_C[ p(\tilde m_p \neq m_p)] + \EE_C[p(\hat m_p \hat m_q \neq m_p m_q | \tilde m_p = m_p)].
\end{align*}
We can bound the first term just as we did for the other protocols. First, use the union bound.
\begin{align*}
  \EE_C[p(\tilde m_p \neq m_p)]  \le \sum_{j=1}^{b-1} \EE_C[p(\tilde p_j \neq p_j | \tilde p_{j-1} = p_{j-1})].
\end{align*}
By~\cref{lem:packing} we can bound each summand as follows:
\begin{align*}
  & \EE_C \left[ p(\tilde p_j \neq p_j | \tilde p_{j-1} = p_{j-1}) \right] \\
  & = \EE_{C} [ \tr[ (I - Q^{(p_j | p_{j-1})}_{B_2}) \rho^{((x_1)_j(p_{j-1}, p_j, q_j ) (x_2)_j (p_{j-1}))}_{B_2} ] ] \\
  & = \EE_{C_{(U)_j (X_2)_j}} [ \tr[ (I - Q^{(p_j | p_{j-1})}_{B_2}) \rho^{((u)_j(p_{j-1}, p_j ) (x_2)_j (p_{j-1}))}_{B_2} ] ] \\
  & \le f(2,\varepsilon) + 4 \sum_{T = \{ P_j\}}  2^{R_p -D_H^\eps(\rho_{(U)_j (X_2)_j (B_2)_j} \Vert \rho_{(U)_j (X_2)_j (B_2)_j}^{(\{X_{S_T}, B_2 \})})} \\
  & = f(2,\varepsilon) + 4 \times  2^{R_p -D_H^\eps(\rho_{U X_2 B_2} \Vert \rho_{U X_2 B_2}^{(\left\{ U , B_2\right\})})},
\end{align*}
where we used $S_{ \left\{ P_j \right\}} = \left\{ (U)_j \right\}$. We dropped the index $j$ in the last equality since $(U)_j (X_2)_j \sim p_{U X_2}$. Hence, overall,
\begin{align*}
  & \EE_C \left[ p(\tilde m_p \neq m_p) \right] \\
  & \le b \left[  f(2,\varepsilon) + 4 \times 2^{R_p-D_H^\eps(\rho_{U X_2 B_2} \Vert \rho_{U  X_2 B_2}^{(\left\{ U, B_2 \right\})})}\right].
\end{align*}
For the second term, we again invoke~\cref{lem:packing} to obtain~\cref{eq:PQmismatches}.
\begin{figure*}[!t]
  \normalsize
\begin{align} 
 & \EE_{C}[p(\hat m_p \hat m_q \neq m_p m_q | \tilde m_p =m_p) ]  = \nonumber\\
  &= \EE_{C} \left[ \tr\left[  (I - Q^{(p_1 \cdots p_{b-1}  q_1 \cdots q_{b-1} )}_{B^b_3}) \bigot_{j=1}^b \rho^{(x_1(p_{j-1}, p_j, q_j) x_2(p_{j-1}))}_{B_3}\right] \right] \nonumber\\
  & = \EE_{C} \left[ \tr\left[  (I - Q^{(p_1\cdots p_{b-1} , q_1 \cdots q_{b-1})}_{B^b_3}) \bigot_{j=1}^b \rho^{(u(p_{j-1}, p_j) x_1(p_{j-1}, p_j, q_j) x_2(p_{j-1}))}_{B_3}\right]\right] \nonumber\\
  & \le f(3b,\varepsilon) + 4 \times \sum_{J_p, J_q \subseteq [b-1]: j_p + j_q > 0}  2^{j_p R_p + j_q R_q  -D_H^\eps \left( \rho_{U^b X^b_1 X^b_2 B^b_3} \Big\Vert \rho_{U^b X^b_1 X^b_2 B^b_3}^{(\{ X_{S_{(J_p, J_q)}}, B_3^b \})}\right) }\label{eq:PQmismatches}.
\end{align}
\hrulefill
\end{figure*}
We defined $S_{(J_p, J_q)} \equiv \left\{  X^{\mathcal{J}}_1  X^{J_p'}_2 U^{\mathcal{J}_p}\right\}$, $\mathcal{J} \equiv \mathcal{J}_p \cup J_q$, $\mathcal{J}_p \equiv J_p \cup J_p'$, $J_p' \equiv \left\{ j \in [b] | j-1 \in J_p \right\}$, and $j_p \equiv |J_p|, j_q \equiv |J_q|$.
Also, note that $\rho_{U^b X_1^b X_2^b B_3^b} = \rho_{U X_1 X_2 B_3}^{\ot b}$. Thus, overall, we have proved the following proposition.
\begin{prp}
  Given $R_p, R_q \in \R_{\ge 0}, \,\varepsilon \in (0,1),\, b \in \N$, the triple $(\frac{b-1}{b} (R_p+R_q), b, \delta)$ is achievable for the classical-quantum relay channel, where
\begin{align*}
  \delta & = b \left[  f(2, \varepsilon)  + 4 \times 2^{R_p-D_H^\eps \left(\rho_{U X_2 B_2} \Big\Vert \rho_{U  X_2 B_2}^{(\left\{ U, B_2 \right\})}\right)}\right] + f(3b,\varepsilon) \nonumber\\
  & + 4 \sum_{J_p, J_q \subseteq [b-1]: j_p + j_q > 0}  \\
  & \quad \quad 2^{j_p R_p + j_q R_q  -D_H^\eps\left(\rho_{U^b X^b_1 X^b_2 B^b_3} \Big\Vert \rho_{U^b X^b_1 X^b_2 B^b_3}^{(\{X^{\mathcal{J}}_1  X^{J_p'}_2 U^{\mathcal{J}_p}, B_3^b \})}\right)}. 
\end{align*}
\end{prp}
In the asymptotic limit, the error vanishes provided
\begin{align}
  R_p < I(U; B_2 | X_2) \label{eq:generalpBound}
\end{align}
and, for all $J_p, J_q \subseteq [b-1]$,
\begin{align}\label{eq:pDecodeAsympt}
  j_p R_p +j_q R_q < I(X^{\mathcal{J}}_1 X^{J_p'}_2 U^{\mathcal{J}_p}; B^b_3| X_1^{\overline{\mathcal{J}}} X_2^{\overline{J_p'}} U^{\overline{\mathcal{J}_p}})_{\rho_{U^b X^b_1 X^b_2 B^b_3}}.
\end{align}
Note $J_p, J_q \subseteq [b-1]$ and $J_p' \subseteq [2:b]$. However, we use the convention that all complementary sets are with respect to largest containing set $[b]$.%
\footnote{The $b$\textsuperscript{th} messages and estimates match, but in general the $b$\textsuperscript{th} $x_1, x_2, u$ depend also on the $(b-1)$\textsuperscript{th} messages and estimates.}
We can simplify~\cref{eq:pDecodeAsympt} via a general lemma:
\begin{lem}\label{lem:rmCond}
  Let $\rho_{B_1 \dots B_m}$ be $m$-partite quantum state. We consider the state $\rho_{B_1 \dots B_m}^{\ot n}$ for some $n \in \N$. Now, let $B, B', C$ be disjoint subsystems of $(B_1 \dots B_m)^{\ot n}$ and such that $B, B'$ are supported on disjoint tensor factors. Then,
  \begin{align*}
    I(B; B' | C) =0.
  \end{align*}
\end{lem}
\begin{proof}
  We prove this by the definition of the conditional mutual information and the fact that $\rho_{B_1\dots B_m}^{\ot n}$ is a tensor product state:
  \begin{align*}
    I(B; B' | C) & = S(B C) + S(B' C) - S(B B' C) - S(C)\\
    & = S(B C_B) + S(C_{\overline B}) + S(B' C_{B'}) + S(C_{\overline{B'}}) \\
    & - S(B C_B) - S(B' C_{B'}) - S(C_{\overline{B B'}}) -S(C)\\
    &=0.
  \end{align*}
  where $C_B$ is the subsystem of $C$ supported on the tensor factors that support $B$ and $C_{\overline B}$ is the rest of $C$.
\end{proof}

Using~\cref{lem:rmCond} this and the chain rule, for any conditional mutual information quantity we can remove conditioning systems which are supported on tensor factors disjoint from those that support the non-conditioning systems. This is key in the following analyses. For instance, in~\cref{eq:pDecodeAsympt}, $\overline{\mathcal{J}}$ and $\mathcal{J} \cup J_p' \cup \mathcal{J}_p = \mathcal{J}$ are supported on disjoint tensor factors, and so we can remove the conditioning on the $X_1^{\ol{\mathcal{J}}}$ system:
\begin{align*}
  & I(X^{\mathcal{J}}_1 X^{J_p'}_2 U^{\mathcal{J}_p}; B^b_3| X_1^{\overline{\mathcal{J}}} X_2^{\overline{J_p'}} U^{\overline{\mathcal{J}_p}}) \\
  & = I(X^{\mathcal{J}}_1 X^{J_p'}_2 U^{\mathcal{J}_p}; B^b_3| X_2^{\overline{J_p'}} U^{\overline{\mathcal{J}_p}}) \\
  & + I(X^{\mathcal{J}}_1 X^{J_p'}_2 U^{\mathcal{J}_p}; X_1^{\ol{ \mathcal{J}}} | B^b_3 X_2^{\overline{J_p'}} U^{\overline{\mathcal{J}_p}}) \\
  & - I(X^{\mathcal{J}}_1 X^{J_p'}_2 U^{\mathcal{J}_p}; X_1^{\ol{ \mathcal{J}}} | X_2^{\overline{J_p'}} U^{\overline{\mathcal{J}_p}})\\
  & = I(X^{\mathcal{J}}_1 X^{J_p'}_2 U^{\mathcal{J}_p}; B^b_3| X_2^{\overline{J_p'}} U^{\overline{\mathcal{J}_p}}).
\end{align*}
Thus,~\cref{eq:pDecodeAsympt} reduces to
\begin{align*}
  j_p R_p +j_q R_q < I(X^{\mathcal{J}}_1 X^{J_p'}_2 U^{\mathcal{J}_p}; B^b_3|  X_2^{\overline{J_p'}} U^{\overline{\mathcal{J}_p}}).
\end{align*}
\begin{figure}[t]
\begin{center}
\includegraphics[scale=0.5]{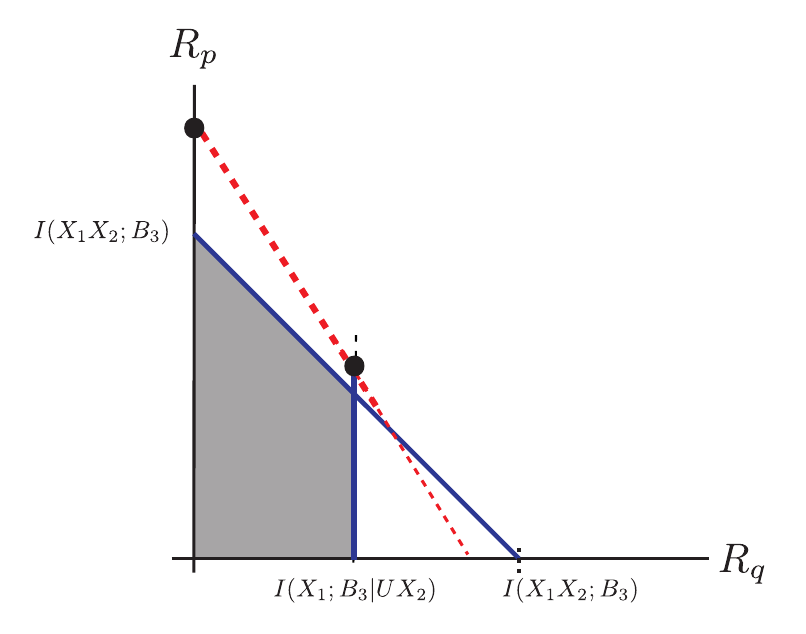}
\end{center}
\caption{$R_p-R_q$ rate region. Gray region is $S$, defined by the blue (solid) lines that correspond to~\cref{eq:qBound} and~\cref{eq:p+qBound}. The red (dashed) line corresponds to~\cref{eq:generalpqBound} for a fixed $J_p, J_q \subseteq [b-1]$.}\label{fig:rates}
\end{figure}

We claim that the set of pairs $(R_p, R_q)$ that satisfy these bounds gives the classical partial decode-forward lower bound with quantum mutual information quantities in the limit of large $b$.%
\footnote{This will also cause $\frac{b-1}{b} \to 1$ so that the rate we achieve really is $R_p+R_q$. }
In particular, we show:
\begin{prp}
  Let
  \begin{align*}
    & S(b) \equiv \Big\{ (R_p, R_q) \in \R^2_{\ge 0} \big\vert  
    \forall J_p, J_q \subseteq [b-1] \\
    & \text{ such that }  j_p +j_q >0, \\
    & j_p R_p +j_q R_q  < I(X^{\mathcal{J}}_1 X^{J_p'}_2 U^{\mathcal{J}_p}; B^b_3|  X_2^{\overline{J_p'}} U^{\overline{\mathcal{J}_p}})_{\rho_{U^b X^b_1 X^b_2 B^b_3}} \Big\}
  \end{align*}
  and
  \begin{align*}
    S \equiv \Big\{ (R_p, R_q) \in \R^2_{\ge 0} | 
    & R_q < I(X_1; B_3 |U X_2)_{\rho_{U X_1 X_2 B_3}}, \\
    & R_p + R_q < I(X_1 X_2 ; B_3)_{\rho_{X_1 X_2 B_3}} \Big\},
  \end{align*}
  where $\rho_{U X_1 X_2 B_2 B_3}$ is given by~\cref{eq:cqStatePartialDecode}. Then, $\lim_{b \to \infty} S(b)$ exists and is equal to $S$.
\end{prp}
Note that the bounds that define $S$ do not match the bounds given for instance in~\cite{elgamal2011network} since we do not first decode $P$ and thereby $Q$, but instead jointly decode to obtain all messages simultaneously. However, in the end we still obtain the same lower bound on the capacity.
\begin{proof}
  For reference, we list the bounds:
  \begin{align}\label{eq:generalpqBound}
    j_p R_p +j_q R_q  & < I(X^{\mathcal{J}}_1 X^{J_p'}_2 U^{\mathcal{J}_p}; B^b_3|  X_2^{\overline{J_p'}} U^{\overline{\mathcal{J}_p}})_{\rho_{U^b X^b_1 X^b_2 B^b_3}}
  \end{align}
  and
  \begin{align}
    R_q < I(X_1;B_3| U X_2)_{\rho_{U X_1 X_2 B_3}} \label{eq:qBound}\\
    R_p + R_q < I(X_1 X_2 ; B_3)_{\rho_{X_1 X_2 B_3}} \label{eq:p+qBound}.
  \end{align}

  We first claim $\limsup_{b \to \infty} S(b) \subseteq S$.
  Consider $J_p, J_q = [b-1]$, in which case~\cref{eq:generalpqBound} becomes
  \begin{align*}
    (b-1)(R_p + R_q) < I(X_1^b (X_2)_2^b U^b ; B_3^b | (X_2)_1),
  \end{align*}
  which, using~\cref{lem:rmCond}, can be manipulated into
  \begin{align*}
    R_p + R_q & < \frac{b}{b-1} I(X_1 X_2 U ; B_3) - \frac{1}{b-1} I( X_2; B_3)\\
    & = I(X_1 X_2 ; B_3) + \frac{1}{b-1} I( X_1; B_3 \vert X_2).
  \end{align*}
  In the limit of large $b$, this becomes~\cref{eq:p+qBound}. To obtain~\cref{eq:qBound}, take $j_p = 0$. Then,~\cref{eq:generalpqBound} becomes by~\cref{lem:rmCond}
  \begin{align*}
    j_q R_q < I(X_1^{J_q} ; B_3^b | X_2^b U^b) = j_q I(X_1 ; B_3 \vert X_2 U).
  \end{align*}
  Now, since $j_p= 0$, $j_q$ cannot be zero, so this is equivalent to
  \begin{align*}
    R_q < I(X_1; B_3 | X_2 U).
  \end{align*}
  The claim thus follows.

  We next claim $S(b) \supseteq S$ for all $b$ and so $\liminf_{b \to \infty} S(b) \supseteq S$. We only need to consider when $j_p > 0$ since otherwise we obtain~\cref{eq:qBound} as shown above, which holds for all $b$. Now, interpret each of the inequalities above as a linear bound on an $R_p$-$R_q$ diagram (see~\cref{fig:rates}). We show that none of the lines corresponding to~\cref{eq:generalpqBound} cuts into $S$. First, fixing $J_p, J_q \subseteq [b-1]$, we find the $R_p$ intercept of said line
  \begin{align*}
    & \frac{1}{j_p} I(X^{\mathcal{J}}_1 X^{J_p'}_2 U^{\mathcal{J}_p}; B^b_3|  X_2^{\overline{J_p'}} U^{\overline{\mathcal{J}_p}}) \\
    &= \frac{1}{j_p} \left( I(X_1^{J_p'} X_2^{J_p'} U^{J_p'} ; B_3^b | X_2^{\overline{J_p'}} U^{\overline{\mathcal{J}_p}} ) +  \cdots \right) \\
    & \ge \frac{1}{j_p} I(X_1^{J_p'} X_2^{J_p'} U^{J_p'} ; B_3^{J_p'} ) \\
    & = I(X_1 X_2 U; B_3) = I(X_1 X_2 ; B_3),
  \end{align*}
  where $\cdots$ stands for some conditional mutual information quantity and therefore is non-negative. Thus, the $R_p$ intercept is at least as large as that of~\cref{eq:p+qBound}, as shown in~\cref{fig:rates}. This determines one of the points of the line.

  We now find another point. We observe that $I(X_1; B_3 | X_2 U) \le I(X_1 X_2 U; B_3)$ so the line associated with~\cref{eq:qBound} intersects that of~\cref{eq:p+qBound} in $\R_{\ge 0}^2$. Hence, it is sufficient to show the bound on $R_p$ when $R_q = I(X_1; B_3 | X_2 U)$ in~\cref{eq:generalpqBound} is weaker than $I(X_1 X_2 U; B_3) - I(X_1; B_3 | X_2 U) = I(X_2 U; B_3)$. To see this, we substitute $R_q = I(X_1; B_3 | X_2 U)$ into~\cref{eq:generalpqBound}:
  \begin{align*}
    & j_p R_p + j_q I(X_1; B_3 | X_2 U) \\
    & \le  I(X^{\mathcal{J}}_1 X^{J_p'}_2 U^{\mathcal{J}_p}; B^b_3|  X_2^{\overline{J_p'}} U^{\overline{\mathcal{J}_p}}) \\
    & = I(X_1^{J_q}; B_3^b | X_1^{\mathcal{J} \backslash J_q} X_2^b U^b) \\
    & + I(X_1^{\mathcal{J} \backslash J_q} X_2^{J_p'} U^{\mathcal{J}_p} ; B_3^b | X_2^{\overline{J_p'}} U^{\overline{\mathcal{J}_p}}) \\
    & = I(X_1^{J_q}; B_3^{J_q} | X_2^{J_q} U^{J_q}) + I(X_2^{J_p'} U^{J_p'} ; B_3^b | X_2^{\overline{J_p'}} U^{\overline{\mathcal{J}_p}}) + \cdots \\
    & = j_q I(X_1; B_3 | X_2 U) + j_p I(X_2 U ; B_3) + \cdots.
  \end{align*}
  This establishes our claim and completes the proof.
\end{proof}
Therefore, combining the bounds~\cref{eq:generalpBound,eq:qBound,eq:p+qBound}, the overall rate $R_p+R_q$ of the entire protocol is achievable if
\begin{align*}
  R_p+R_q < \min\{ &I(X_1 ; B_3 | U X_2)_\rho + I(U; B_2|X_2)_\rho, \\
  & I(X_1 X_2; B_3)_\rho\}. 
\end{align*}
This is sufficient since if it holds we can choose $R_p, R_q$ to satisfy the bounds. It is also necessary since if it is violated, then one of the bounds has to be violated. We can optimize over $p_{U X_1 X_2}$, so we obtain the partial decode-forward lower bound:
\begin{align}
  \label{eq:partial}
  C \ge \max_{p_{U X_1 X_2}} \min\{& I(X_1 ; B_3 | U X_2)_\rho+ I(U; B_2|X_2)_\rho\nonumber \\
  & , I(X_1 X_2; B_3)_\rho\}.
\end{align}

\begin{remark*}
  This coding scheme is optimal in the case when $\mathcal{N}_{X_1 X_2 \to B_2 B_3}$ is \emph{semideterministic}, namely $B_2$ is classical and $\rho_{B_2}^{(x_1 x_2)}$ is pure for all $x_1, x_2$. This is because in this case the partial decode-forward lower bound~\cref{eq:partial} with $U = B_2$ as random variables matches the cutset upper bound~\cref{eq:cutset}. This is possible because of the purity condition, which essentially means $B_2$ is a deterministic function of $X_1, X_2$. The semideterministic classical relay channel was defined and analyzed in~\cite{gamal1982capacity}.
\end{remark*}

\section{Proof of the Quantum Multiparty Packing Lemma}\label{sec:proofpackinglemma}
In this section we prove~\cref{lem:packing} via Sen's joint typicality lemma~\cite{senInPrep}. We then use~\cref{lem:packing} to prove the asymptotic version,~\cref{lem:packingAsympt}. We shall state a special case of the joint typicality lemma, the $t=1$ intersection case in the notation of~\cite{senInPrep}, as a theorem. For the sake of conciseness, we suppress some of the detailed expressions.

We first give some definitions. A \emph{subpartition} $\mathcal L$ of some set $S$ is a collection of nonempty, pairwise disjoint subsets of $S$.
We define $\bigcup(\mathcal L)$ to be their union, that is, $\bigcup(\mathcal L)\equiv\bigcup_{L\in\mathcal L} L$.
Note that $\bigcup (\mathcal{L}) \subseteq S$.
We say a subpartition $\mathcal{L}$ of $S$ \emph{covers} $T \subseteq S$ if $T \subseteq \bigcup (\mathcal{L})$.
\begin{thm}[One-shot Quantum Joint Typicality Lemma~\cite{senInPrep}]\label{thm:conditional}
  Let
  \begin{align*}
    \rho_{XA} = \sum_{x}^{} p_X(x) \state{x}_X \ot \rho_A^{(x)}
  \end{align*}
  be a classical-quantum state where $A \equiv A_1 \dots A_N$ and $X \equiv X_1 \dots X_M$.
  Let $\eps \in (0,1)$ and let $Y=Y_1\dots{}Y_{N+M}$ consist of $N+M$ identical copies of some classical system, with total dimension $d_Y$.
  Then there exist quantum systems $\wh A_k$ and isometries $\wh J_k\colon A_k \to \wh A_k$ for $k\in[N]$, as well as a cqc-state of the form
  \begin{align*}
    \wh \rho_{X \wh A Y} = \frac{1}{d_Y} \sum_{x,y}^{} p_X(x) \state{x}_X \ot \wh \rho_{\wh A}^{(x,y)} \ot \state{y}_Y,
  \end{align*}
  and a cqc-POVM $\wh \Pi_{X\wh A Y }$, such that, with $\wh J \equiv  \bigot_{k \in [N]} \wh J_k$,
\begin{enumerate}
\item\label{it:thm conditional 1}
$\left\Vert \wh\rho_{X\wh AY} - (\id_{X} \ot \wh J) \rho_{XA}  (\id_{X} \ot \wh J)^\dagger \ot \tau_{Y} \right\Vert_1 \\$
$\leq f(N,M, \eps)$, where $\tau_Y=\frac1{d_Y}\sum_y \ket y\bra y_Y$ denotes the maximally mixed state on $Y$,
\item\label{it:thm conditional 2}
$\tr\left[\wh \Pi_{X \wh AY} \wh\rho_{X\wh A Y} \right] \geq 1 - g(N,M,\eps)$.
\item\label{it:thm conditional 3}
Let $\mathcal{L}$ be a subpartition of $[M] \sqcup [N]$ that covers $[N]$.
Define $Y_{\mathcal L} := Y_{\bigcup(\mathcal L)}$, $S \equiv [M] \cap \bigcup (\mathcal{L})$, $\ol{S} \equiv [M] \setminus S$ and the ``conditional'' quantum states
  \begin{align*}
    \wh \rho^{(x_{\ol{S}}, y_{\ol{S}})}_{X_{S} \wh A Y_\mathcal{L} } \equiv \frac{1}{d_{Y_{\mathcal{L}}}} \sum_{x_{S}, y_\mathcal{L}}^{} & p_{X_{S} | X_{\ol{S}}} (x_{S} | x_{\ol{S}}) \state{x_{S}}_{X_{S}} \\
    & \ot \wh \rho_{\wh A}^{(x_{\ol{S}} x_{S}, y_{\ol{S}} y_{\mathcal{L}})} \ot \state{y_\mathcal{L}}_{Y_\mathcal{L}}
    \end{align*}
    \begin{align*}
    \rho^{(x_{\ol{S}})}_{X_{S}  A } & \equiv \sum_{x_{S}}^{} p_{X_{S} | X_{\ol{S}}}(x_{S} | x_{\ol{S}}) \state{x_{S}}_{X_{S}} \ot \rho_{A}^{(x_{\ol{S}} x_{S})}.
  \end{align*}
  We can now define
  \begin{align*}
    \wh \rho_{X \wh A Y}^{(\mathcal{L})} \equiv \frac{1}{d_{Y_\ol{S}}} \sum_{x_\ol{S}, y_\ol{S}}^{} & p_{X_\ol{S}}(x_{\ol{S}}) \state{x_{\ol{S}}}_{X_\ol{S}} \ot \\
    & \bigot_{L \in \mathcal{L}} \wh \rho^{(x_{\ol{S}}, y_{\ol{S}})}_{X_{L \cap [M]} \wh A_{L \cap [N]}  Y_L } \ot \state{y_\ol{S}}_{Y_\ol{S}}
    \end{align*}
    \begin{align*}
    \rho^{(\mathcal{L})}_{XA} & \equiv \sum_{x_\ol{S}}^{} p_{X_\ol{S}}(x_{\ol{S}}) \state{x_{\ol{S}}}_{X_\ol{S}} \ot \bigot_{L \in \mathcal{L}} \rho^{(x_{\ol{S}})}_{X_{L \cap [M]} A_{L \cap [N]} }
  \end{align*}
  in terms of the reduced density matrices of the states
  $\wh \rho^{(x_{\ol{S}}, y_{\ol{S}})}_{X_{S} \wh A Y_\mathcal{L} }$
  and
  $\rho^{(x_{\ol{S}})}_{X_{S}  A }$
  defined above.
  Then, 
  \begin{align*}
    & \tr\left[\wh \Pi_{X \wh AY} \left(  \wh \rho_{X \wh A Y}^{(\mathcal{L})}\right) \right] \\
    & \leq 2^{-D_H^\eps\left(\rho_{XA} \Vert \rho^{(\mathcal{L})}_{XA} \right)} + h(N,M,d_A,d_Y).
  \end{align*}
\end{enumerate}
Here, $f(N,M,\eps)$, $g(N,M,\eps)$, $h(N,M,d_A,d_Y)$ are universal functions (independent of the setup) such that
\begin{align*}
  \lim_{\eps\to0} f(N,M,\eps) & = \lim_{\eps\to0} g(N,M,\eps)\\
  & =\lim_{d_Y \to \infty} h(N,M,d_A,d_Y)\\
  & =0.
\end{align*}
\begin{align*}
\end{align*}
\end{thm}
\begin{proof}
This follows readily from Sen's Lemma~1 in~\cite{senInPrep} with an appropriate change of notation and suitable simplifications. We use Sen's terminology and notation.
We choose
$k_{\text{Sen}} \equiv N$,
$c_{\text{Sen}} \equiv M$,
$\mathcal L_{\text{Sen}}$ a system isomorphic to our $Y_k$,
$\delta_{\text{Sen}}=\eps^{1/3N}$,
and the same error $\eps$ for each pseudosubpartition of $[M] \sqcup [N]$.
We denote $\wh A_k \equiv (A''_k)_{\text{Sen}}$, so that $(A'_k)_{\text{Sen}} = \wh A_k Y_k$ and $(X'_k)_{\text{Sen}} = X_k Y_k$; that is, we explicitly include the augmenting systems in our notation.
We also write $\hat J_k$ for the natural embedding $A_k \hookrightarrow A''_k$.
Then Sen's lemma yields a state $\wh\rho_{X\wh AY} \equiv \rho'_{\text{Sen}}$ and a POVM $\wh\Pi_{X\wh AY} \equiv \Pi'_{\text{Sen}}$ that satisfies all desired properties.
First, Statement~1 in Sen's lemma asserts that $\wh\rho_{X\wh AY}$ and $\wh\Pi_{X\wh AY}$ are cqc.
Next, our properties~\ref{it:thm conditional 1} and \ref{it:thm conditional 2} are direct restatements of his statements~2 and 3, with $f(N,M,\eps) = 2^{(N+M)/2+1}\eps^{1/3N}$ and $g(N,M,\eps) = 2^{2^{MN+4} (N+1)^N} 2^{(N+M)^2} \eps^{1/3} + 2^{(N+M)/2+1} \eps^{1/3N}$.
Finally, we apply statement~4 in Sen's lemma to a subpartition $\mathcal L$ covering $[N]$ and the probability distribution $q_{\text{Sen}}(x) = p_{X_{\ol S}}(x_{\ol S}) \prod_{L\in\mathcal L} p_{X_{L \cap [M]}|X_{\ol S}}(x_{L\cap [M]}|x_{\ol S})$.
Then our $\rho^{(\mathcal L)}_{XA}$ is Sen's $\rho_{(S_1,\dots,S_l)}$ and our $\wh\rho^{(\mathcal L)}_{X\wh AY}$ is Sen's $\rho'_{(S_1,\dots,S_l)}$, so we obtain property~\ref{it:thm conditional 3} with $h(N,M,d_A,d_Y) = 3 \, 2^N d_A d_Y^{-1/2(N+M)}$.
\end{proof}

Now, we prove a lemma that abstractly expresses sufficient properties a multiparty encoding needs to satisfy so that we can use~\cref{thm:conditional} to obtain a simultaneous decoder. We then prove~\cref{lem:packing} by showing that the random codebook generated by a multiplex Bayesian network, that is, a Markov encoding, satisfies such properties. Thus, this lemma is a generalization of~\cref{lem:packing}. We use the notation $X \equiv X_1 \dots X_k$ to denote set of $k \in \N$ systems.

\begin{lem}\label{lem:packingGen}
  Let $\{ p_{X}, \rho_{B}^{(x)}\}$ be an ensemble of quantum states, where $X \equiv X_1 \dots X_k$ with $k \in \N$, $\mathcal I = \mathcal{I}_1 \times \mathcal{I}_2$ an index set, and $\varepsilon \in (0,1)$ a small parameter.
  Now, let $\mathcal{C} =\left\{ x(i) \right\}_{ i \in \mathcal{I} }$ be a family of random variables such that for every $i \in \mathcal{I}$, $x(i) \sim p_{X_1 \cdots X_k}$, and there exists a map%
  \footnote{Note that the bound does not depend on the specific choice of the map.}
  $\Psi: \mathcal{I} \times \mathcal{I} \to \mathcal{P}([k])$ such that for every $i, i' \in \mathcal{I}$, letting $T \equiv \Psi(i,i')$,
  \begin{enumerate}
    \item\label{it:cond1} $x_\ol{T}(i) = x_\ol{T}(i')$ as random variables
    \item\label{it:cond2} $x_{T}(i), x_{T}(i')$ are independent conditioned on $x_{\ol{T}}(i)$ ($= x_{\ol{T}}(i')$),
  \end{enumerate}
  where $\ol T \equiv [k] \setminus T$. Then, for each $i_1 \in \mathcal{I}_1$ there exists a POVM $\{ Q^{(i_2 | i_1)}_B \}_{i_2 \in \mathcal{I}_2}$ dependent on the random variables in $\mathcal{C}$ such that for all $i =(i_1, i_2) \in \mathcal{I}$,
  \begin{align*}
    & \EE_\mathcal{C}\left[ \tr[(I - Q_B^{(i_2 | i_1)}) \rho^{(x(i_1, i_2))}_B ] \right] \\
    & \le  f(k, \varepsilon) + 4 \sum_{i_2' \neq i_2}^{} 2^{-D_H^\varepsilon(\rho_{X B} \Vert \rho_{X B}^{(\{X_S, B \})})},
  \end{align*}
  where $\EE_\mathcal{C}$ is the expectation over the random variables in $\mathcal{C}$, $S \equiv \Psi((i_1,i_2), (i_1,i_2'))$, and
  \begin{align*}
    \rho_{X B} \equiv \sum_{x} p_{X}(x) \state{x}_{X} \otimes \rho^{(x)}_B.
  \end{align*}
  Furthermore, $f(k,\eps)$ is a universal function such that $\lim_{\varepsilon \to 0} f(k,\varepsilon) = 0$.
\end{lem}

Before we prove~\cref{lem:packingGen}, we first show that~\cref{lem:packing} follows from~\cref{lem:packingGen}.

\begin{proof}[Proof of~\cref{lem:packing}]
  Fix subgraph $H$, $\{ \rho_B^{(x_H)} \}_{x_H \in \mathcal{X}_H}$, $D \subseteq J_H$, $\varepsilon \in (0,1)$.
  We invoke~\cref{lem:packingGen} with the ensemble $\{ p_{X_H}, \rho_B^{(x_H)} \}$ with $k = \abs{V_H}$, $\mathcal{I}_1 = M_{\ol D}$, $\mathcal{I}_2 = M_D$, the same $\varepsilon$, and the family of random variables $\mathcal{C} = C_H$.
  We thus identify $\mathcal{I}= M_H = M_D \times M_{\ol D}$. We also define an arbitrary ordering on $V_H$ such that we can identify it with $[k]$.

  We check that $C_H$ satisfies the required properties using the observations we made regarding \algref. First, for every $m_H \in M_H$, $x_H(m_H) \sim p_{X_H}$ by observation~\ref{it:obs1} on p.~\pageref{it:obs1}.

  Next, we claim the map
  \begin{align*}
    & \Psi(m_H, m_H') \\
    & \equiv \left\{ v \in V_H \;\vert\; \exists j \in \ind(v) \text{ such that } (m_H)_j \neq (m'_H)_j \right\}
  \end{align*}
  satisfies the required conditions.
  Let $m_H, m_H' \in M_H$ and $T = \Psi(m_H, m_H')$.
  By definition, given $v \in \ol T$, for all $j \in \ind(v)$, $(m_H)_j = (m_H')_j$.
  Hence, $m_H \vert_{\ind(v)} = m_H' \vert_{\ind(v)}$, so by observation~\ref{it:obs2} on p.~\pageref{it:obs2}, $x_v(m_H) = x_v(m'_H)$ as random variables.
  Thus, $x_{\ol T}(m_H) = x_{\ol T}(m'_H)$ as random variables, so we have established condition~\ref{it:cond1}.

  We now prove the conditional independence statement in condition~\ref{it:cond2} is satisfied.
  For $\xi_{\ol T} \in \mathcal X_{\ol T}$, observation~\ref{it:obs1} shows that
  \begin{align*}
    \Pr\left(x_{\ol T}(m_H) = \xi_{\ol T}\right) = \prod_{v\in \ol T} p_{X_v | X_{\pa(v)}}\left(\xi_v | \xi_{\pa(v)}\right),
  \end{align*}
  where we used that $\pa(\bar T) \subseteq \bar T$ as a consequence of~\cref{eq:ind condition}.
  Next, observation~\ref{it:obs3} implies that the joint distribution of $x_T(m_H)$, $x_T(m'_H)$, and $x_{\ol T}(m_H)$ is given as follows.
  For $\xi, \xi' \in \mathcal X$ such that~$\xi_{\ol T} = \xi'_{\ol T}$,
  \begin{align*}
    &\Pr\left(x_T(m_H) = \xi_T, x_T(m'_H) = \xi'_T, x_{\ol T}(m_H) = \xi_{\ol T}\right)\\
    & =\Pr\left(x(m_H) = \xi, x(m'_H) = \xi'\right)
  \\ &= \left( \prod_{v\in \ol T} p_{X_v|X_{\pa(v)}}(\xi_v|\xi_{\pa(v)}) \right)
  \left( \prod_{v\in T} p_{X_v|X_{\pa(v)}}(\xi_v|\xi_{\pa(v)}) \right)\\
  & \quad \, \left( \prod_{v\in T} p_{X_v|X_{\pa(v)}}(\xi'_v|\xi'_{\pa(v)}) \right).
  \end{align*}
  Hence, $x_T(m_H)$ and $x_T(m'_H)$ are independent conditional on $x_{\ol T}(m_H)$.
 ~\cref{lem:packing} in the form given in~\cref{eq:packing2} then directly follows from applying~\cref{lem:packingGen}.
\end{proof}

Next, we prove that~\cref{lem:packingAsympt} follows from~\cref{lem:packing}.
\begin{proof}[Proof of~\cref{lem:packingAsympt}]
  This follows from~\cref{lem:packing} by replacing $X$ with $n \in \N$ i.i.d.\ copies of itself, $X^n$. Then, associating each $v \in V$ with $X_v^n$, $(G, X^n, M, \ind)$ is a multiplex Bayesian network.

  We now apply \algref~with $(G, X^n, M, \ind)$ as input. This is equivalent to applying it with $(G,X,M,\ind)$ $n$ times. Then, applying~\cref{lem:packing} with inputs $H$, $\{ \bigot_{i=1}^n \rho_{B_i}^{(x_{i,H})} \}_{x_H^n \in \mathcal{X}^n_H}$, $D$, $\varepsilon(n) \in (0,1)$, we obtain a POVM $\{ Q_{B^n}^{(m_D \vert m_{\ol D})} \}_{m_D \in M_D}$ for each $m_{\ol D} \in M_{\ol D}$ such that, for $(m_D, m_{\ol D}) \in M_H$,
  \begin{align*}
    &\quad \EE_{C_H^n}\left[ \tr\left[ (I-Q_{B^n}^{(m_D \vert m_{\ol D})}) \bigot_{i=1}^n \rho_{B_i}^{((x_i)_H(m_D, m_{\ol D}))} \right] \right] \\
    &\leq f(\abs{V_H}, \varepsilon(n)) \\
    & \quad + 4 \sum_{\O \neq T \subseteq D}^{} 2^{(\sum_{t \in T}^{} R_t) - D_H^{\varepsilon(n)}\left(\rho_{X_H^n  B^n} \Big\Vert \rho_{X_H^n B^n}^{(\{X^n_{S_T}, B^n \})}\right)}.
  \end{align*}
  Consider now
  \begin{align*}
    \rho_{X_H^{n} B^n} & = \sum_{x_H^n}^{} p_{X_H}^{\ot n}(x_H^n) \state{x_H^n}_{X_H^n} \ot \bigot_{i=1}^n \rho_{B_i}^{((x_i)_H)}
  \end{align*}
  and
  \begin{align*}
    \rho_{X_H^n B^n}^{(\{X^n_{S_T}, B^n \})} = \sum_{x_H^n}^{} p_{X_H}^{\ot n}(x_H^n) \state{x^n_H}_{X_H^n} \ot \rho_{B^n}^{( x^n_{\ol{S_T}})}.
  \end{align*}
  It is not difficult to see that
  \begin{align*}
    \rho_{X_H^n B^n} & = \left( \sum_{x_H}^{} p_{X_H}(x_H) \state{x_H}_{X_H} \ot \rho_B^{(x_H)} \right)^{\ot n} \\
    & = \rho_{X_H B}^{\ot n},
  \end{align*}
  which conveniently justifies this slight abuse of notation. Furthermore, considering
  \begin{align*}
    \rho_{B^n}^{(x_{\ol{S_T}}^n)} &= \sum_{x_{S_T}^n}^{} p_{X_{S_T} | X_{\ol{S_T}}}^{\ot n}(x_{S_T}^n | x_{\ol{S_T}}^n) \bigot_{i=1}^n \rho_{B_i}^{( (x_i)_{S_T} (x_i)_{\ol{S_T}})} \\
    & = \bigot_{i=1}^n \rho_{B_i}^{(x_i)_{\ol{S_T}}},
  \end{align*}
  we likewise conclude
  \begin{align*}
    \rho_{X_H^{n} B^n}^{(\{X^n_{S_T}, B^n \})} =  \left(\rho_{X_H B}^{(\{X_{S_T}, B \})}\right)^{\ot n}.
  \end{align*}
  The conclusion therefore follows by~\cref{eq:CMI} where we choose $\eps(n), \delta(n)$ such that $\varepsilon(n) \to 0$ so that $f(\abs{V_H}, \varepsilon(n)) \to 0$ but also for all $\O \neq T \subseteq D$,
  \begin{align*}
    2^{n(\sum_{t \in T}^{} R_t) - D_H^{\varepsilon(n)}(\rho_{X_H  B}^{\ot n} \Vert (\rho_{X_H B}^{(\{X_{S_T}, B \})} )^{\ot n})} \to 0
  \end{align*}
  when the rate inequalities are satisfied. Given~\cref{eq:DHepsBound}, one possibility is $\varepsilon(n) = 1/n$ and $\delta(n) = n^{-1/4}$. This concludes the proof.
\end{proof}
Finally, we prove~\cref{lem:packingGen}. Note that~\cref{thm:conditional} gives a pair $\wh \rho, \wh \Pi$ that satisfy joint typicality properties but live in a larger Hilbert space. In order to prove~\cref{lem:packing}, which claims the existence of a POVM on the original Hilbert space, we need to construct the corresponding POVM in the larger Hilbert space and then appropriately invert the isometry. 
There is also an extra classical system $Y$ associated with the $X$ systems, which we can interpret as an additional random codebook. We use a conventional derandomization argument to eliminate it from the statement. The extra $Y$'s associated with the $B$ systems we simply trace over.
\begin{proof}[Proof of~\cref{lem:packingGen}]
  We invoke~\cref{thm:conditional} with inputs the $\rho_{X B}$, $\eps$, and a classical system $Y Z$. Here $X \equiv X_1 \dots X_k$, $Y \equiv Y_1 \dots Y_k$ and $Z$ is a classical system associated with $B$, to obtain a quantum state $\wh \rho_{X \wh B Y Z}$ and POVM $\wh \Pi_{X \wh B Y Z}$ which we can expand as follows:
  \begin{align*}
    \wh \rho_{X \wh B Y Z} = \bigoplus_{x, y} p_{X}(x) \state{x}_{X} \otimes \frac{1}{d_{Y}} \state{y}_{Y} \ot \wh \rho_{\wh B Z}^{(x, y)}
  \end{align*}
  \begin{align*}
    \wh \Pi_{X \wh B Y Z} = \bigoplus_{x,  y} \state{x}_{X} \otimes \state{y}_{Y} \ot \wh \Pi_{\wh B Z}^{(x, y)}.
  \end{align*}
  Now, for every $x_j \in \mathcal X_j$, draw $y_j(x_j)$ uniformly at random from~$\mathcal Y_j$, and consider the random vectors~$y(x) := (y_1(x_1),\dots,y_k(x_k))$.
  We use these random vectors and the codebook $\mathcal C = \{ x(i) \}_{i\in\mathcal I}$ to define a codebook $\mathcal{C}' = \{ y(i) \}_{i\in\mathcal I}$, where we set $y(i) = y(x(i))$.
  We also define the joint codebook~$\mathcal{C}'' = \{ x(i)y(i) \}_{i \in \mathcal I}$.
  Then, for every $i, i' \in \mathcal{I}$, letting $T \equiv \Psi(i,i')$, the following holds:
  \begin{enumerate}
    \item $x_\ol{T}(i) y_\ol{T}(i)  = x_\ol{T}(i') y_\ol{T}(i') $ as random variables,
    \item $x_{T}(i) y_{T}(i)$ and $x_{T}(i') y_{T}(i')$ are independent conditioned on $x_\ol{T}(i) y_\ol{T}(i) $ ($=x_\ol{T}(i') y_\ol{T}(i')$),
  \end{enumerate}
  with probabilities
  \begin{align*}
    &p_{X_\ol{T} Y_\ol{T}}(x_\ol{T}, y_\ol{T})  =    p_{X_\ol{T} }(x_\ol{T} ) \cdot p_{Y_\ol{T}}(y_\ol{T}) =  \frac{1}{d_{Y_\ol{T}}}  p_{X_\ol{T} }(x_\ol{T} ) \\
    &p_{X_{T} Y_{T}| X_\ol{T} Y_\ol{T}}(x_{T}, y_{T} |x_\ol{T}, y_\ol{T}) =   \frac{1}{d_{Y_{T} }} p_{X_{T} |X_\ol{T} }(x_{T}|x_\ol{T} ).
  \end{align*}
  Define the indexed objects:
  \begin{align*}
    \wh \rho^{(i)}_{\wh B Z} \equiv \wh \rho_{\wh B Z}^{(x(i), y(i))} \quad \text{and} \quad
    \wh \Pi^{(i)}_{\wh B Z} \equiv \wh \Pi_{\wh B Z}^{(x(i), y(i))}.
  \end{align*}
  We then define the square-root measurement
  \begin{align*}
    \wh Q_{\wh B Z}^{(i_2| i_1)} \equiv \left( \sum_{i_2' \in \mathcal{I}_2}^{} \wh \Pi^{(i_1, i_2')}_{\wh B Z} \right)^{-1/2} \wh \Pi_{\wh B Z}^{(i_1, i_2)} \left( \sum_{i_2' \in \mathcal{I}_2}^{} \wh \Pi^{(i_1, i_2')}_{\wh B Z} \right)^{-1/2}
  \end{align*}
  and ``invert'' the isometry $\wh J$ to obtain the following family of POVM's on the original Hilbert space:
  \begin{align*}
    Q_B^{(i)} = Q_B^{(i_2|i_1)} \equiv \frac{1}{d_Z} (\wh J_{B \to \wh B})^\dagger \tr_Z\left[ \wh Q_{\wh B Z}^{(i)} \right] \wh J_{B \to \wh B}.
  \end{align*}
  Note that we have a POVM for each value of $i_1$ and these POVM's are dependent on our random encoding $x(i)$ \emph{and} random choice of $y(i)$.

  Now, fixing $i = (i_1, i_2) \in \mathcal{I}$, we compute the probability of error averaged over the random choice of $x(i)$ \emph{and} $y(i)$, denoting this by $\EE \equiv \EE_{\mathcal{C}''}$:
  \begin{align*}
    & \EE\tr\left[\left(I- Q^{(i)}_B\right) \rho^{(i)}_B\right] \\
    & = 1- \EE \tr\left[ Q_B^{(i)} \rho^{( i)}_B \right]\\
    & = 1- \EE \tr\left[ \wh Q^{(i)}_{\wh B Z} \left( \wh J_{B \to \wh B} \rho_B^{( i)} \wh J_{B \to \wh B}^\dag \ot \tau_Z \right)\right]\\
    & \le 1- \EE \tr\left[ \wh Q^{(i)}_{\wh B Z} \wh \rho_{\wh B Z}^{( i)} \right] + \EE \norm{ \wh J_{B \to \wh B} \rho^{( i)}_B \wh J_{B \to \wh B}^\dag \ot \tau_Z - \wh \rho_{\wh B Z}^{( i)} }_1\\
    & = 1- \EE \tr\left[ \wh Q^{(i)}_{\wh B Z} \wh \rho_{\wh B Z}^{( i)} \right] + \\
    & \norm{ \left(\id_{X B} \ot \wh J_{B \to \wh B}\right) \rho_{X B} \left(\id_{X B} \ot \wh J_{B \to \wh B}^\dag\right) \ot \tau_{Y Z} - \wh \rho_{X \wh B Y Z} }_1\\
    & \le 1- \EE \tr\left[ \wh Q^{(i)}_{\wh B Z} \wh \rho_{\wh B Z}^{( i)} \right] + f(1, k,\eps)\\
    & \le 2 \left( 1- \EE \tr\left[ \wh \Pi^{(i)}_{\wh B Z} \wh \rho_{\wh B Z}^{( i)} \right] \right) + 4 \sum_{i_2' \neq i_2}^{} \EE \tr\left[  \wh \Pi^{(i_1, i_2')}_{\wh B Z} \wh \rho_{\wh B Z}^{(i_1, i_2)}\right] \\
    & \quad + f(1, k,\eps)\\
    & \le 4 \sum_{i_2' \neq i_2}^{} \EE \tr\left[  \wh \Pi^{(i_1, i_2')}_{\wh B Z} \wh \rho_{\wh B Z}^{( i_1, i_2)}\right] + f(1, k,\eps) + 2 g(1, k, \eps) ,
  \end{align*}
  where in the last three inequalities we used~\cref{thm:conditional} and the Hayashi-Nagaoka lemma~\cite{hayashi2003general,wilde2013quantum}.

  We consider the first term. Let $S = \Psi( (i_1,i_2), (i_1, i_2'))$. Note that by our conditions on the random codebook, the codewords are equal as random variables on $\ol{S}$ and hence we obtain~\cref{eq:codebookConditions}.
  \begin{figure*}[!t]
    \normalsize
    \begin{align} \label{eq:codebookConditions}
    & 4 \sum_{i_2' \neq i_2}^{} \EE \tr\left[  \wh \Pi^{(i_1, i_2')}_{\wh B Z} \wh \rho_{\wh B Z}^{( i_1,i_2)}\right] \nonumber\\
    &  = 4 \sum_{i_2' \neq i_2}^{} \EE_{X X^{\prime} Y Y^{\prime}} \tr\left[  \wh \Pi^{(i_1, i_2')}_{\wh B Z} \wh \rho_{\wh B Z}^{( i_1,i_2)}\right]\nonumber\\
    &  = 4 \sum_{i_2' \neq i_2}^{} \tr\left[  \EE_{X_\ol{S} Y_\ol{S}} \left[ \EE_{X_S' Y_S' \vert X_\ol{S} Y_\ol{S}} \left(\wh \Pi^{(i_1, i_2')}_{\wh B Z}\right) \EE_{X_S Y_S | X_\ol{S} Y_\ol{S}} \left( \wh \rho_{\wh B Z}^{( i_1, i_2)}\right) \right] \right] \nonumber\\
    &  = 4  \sum_{i_2' \neq i_2}^{} \tr \left[ \sum_{x_\ol{S}, y_\ol{S}} p(x_\ol{S}) \frac{1}{d_{Y_\ol{S}}}  \sum_{x_S', y_S'}^{} p(x_S' | x_\ol{S}) \frac{1}{d_{Y_S}} \wh \Pi^{(x_S' x_\ol{S}, y_S' y_\ol{S})}_{\wh B Z} \sum_{x_S, y_S}^{} p(x_S| x_\ol{S}) \frac{1}{d_{Y_S}} \wh \rho_{\wh B Z}^{( x_S x_\ol{S}, y_S y_\ol{S})} \right]\nonumber\\
    &  = 4  \sum_{i_2' \neq i_2}^{} \tr \left[ \sum_{x_\ol{S}, y_\ol{S}} p(x_\ol{S}) \frac{1}{d_{Y_\ol{S}}}  \sum_{x_S', y_S'}^{} p(x_S' | x_\ol{S}) \frac{1}{d_{Y_S}} \wh \Pi^{(x_S' x_\ol{S}, y_S' y_\ol{S})}_{\wh B Z} \wh \rho_{\wh B Z}^{( x_\ol{S}, y_\ol{S})} \right]\nonumber\\
    &  = 4  \sum_{i_2' \neq i_2}^{} \tr \left[ \sum_{x, y} p(x) \frac{1}{d_{Y}} \wh \Pi^{(x, y)}_{\wh B Z} \wh \rho_{\wh B Z}^{( x_\ol{S}, y_\ol{S})} \right]\nonumber\\
    &  = 4 \sum_{i_2' \neq i_2}^{} \tr\left[  \wh \Pi_{X \wh B Y Z} \wh \rho_{X \wh B Y Z}^{(\{X_{S} Y_S, \hat{B}Z \})} \right] \nonumber\\
    & \le 4 \sum_{i_2' \neq i_2}^{} 2^{-D_H^\eps(\rho_{X B} \Vert \rho_{X B}^{(\{X_{S}, B \})})} +\eps^{1/3}.
  \end{align}
  \hrulefill
\end{figure*}
  In the first two equalities we use the notation $X \equiv x(i_1, i_2), X^{\prime} \equiv x(i_1, i_2')$ and similarly for $Y, Y^{\prime}$. In the fourth equality $\wh \rho_{\hat B Z}^{(x_\ol{S}, y_\ol{S})}$ is the marginal of the conditional density operator $\wh \rho_{X_S \wh B Y_S Z}^{(x_\ol{S}, y_\ol{S})}$.
  In the last inequality we use~\cref{thm:conditional} and choose the dimensions of $Y, Z$ to be sufficiently large so that $h(1,k, d_B, d_Y d_Z) \le \varepsilon^{1/3}$.

  Finally, we can invoke the usual derandomization argument to remove the dependency of our POVM on the choice of $y(i)$. That is, we know that
  \begin{align*}
    & \EE \tr\left[ (I- Q^{(i)}_B) \rho_B^{(i)} \right] \\
    & = \EE_{\mathcal{C}'} \EE_\mathcal{C} \tr\left[ (I- Q^{(i)}_B) \rho_B^{(i)} \right] \\
    & \le \eps^{1/3} + f(1, k, \eps) + 2 g(1, k, \eps)+ 4 \sum_{i_2' \neq i_2}^{} 2^{-D_H^\eps(\rho_{X B} \Vert \rho_{X B}^{(\{X_{S}, B \})})}.
  \end{align*}
  Hence, there is a particular choice of $y(i)$ such that the corresponding POVM $Q^{(i_2|i_1)}_B$ satisfies the bound in
 ~\cref{lem:packingGen}, with
  \begin{align*}
    f(k, \varepsilon)  = \eps^{1/3} + f(1, k, \eps) + 2 g(1, k, \eps).
  \end{align*}
\end{proof}
Finally, we quickly derive~\cref{eq:explicitF} using the definitions of $f(N, M,\varepsilon)$ and $g(N,M,\varepsilon)$ from the proof of~\cref{thm:conditional}:
\begin{align*}
  f(k, \varepsilon) & = \varepsilon^{1/3} + 2^{(k+1)/2 +1} \varepsilon^{1/3} \\
  & \quad + 2 \left( 2^{2^{k+4} \times 2} 2^{(k+1)^2} \varepsilon^{1/3} + 2^{(k+1)/2+1} \varepsilon^{1/3}\right) \nonumber \\
  & =  \left( 1+ 6 \times 2^{\frac{k+1}{2}} + 4 \times 2^{2^{k+5} + k^2 + 2k} \right) \varepsilon^{1/3}.
\end{align*}

\section{Conclusions}\label{sec:conclusions}
The packing lemma is a cornerstone of classical network information theory, used as a black box in the analyses of all kinds of network communication protocols. At its core, the packing lemma follows from properties of the set of jointly typical sequences for multiple random variables. In this letter, we provide an analogous statement in the quantum setting that we believe can serve a similar purpose for quantum network information theory. We illustrate this by using it as a black box to prove achievability results for the classical-quantum relay channel. Our result is based on a joint typicality lemma recently proved by Sen~\cite{senInPrep}. This result, at a high level, provides a single POVM which achieves the hypothesis testing bound for all possible divisions of a multiparty state into a tensor product of its marginals. This result allows for the construction of finite blocklength protocols for quantum multiple access, relay, broadcast, and interference channels~\cite{senInPrep2}.

Two alternative formulations of joint typicality were proposed in~\cite{dutil2011multiparty} and~\cite{drescher2013simultaneous}. In the first work, the author conjectured the existence of the jointly typical state that is close to an i.i.d.\ multiparty state but with marginals whose purities satisfy certain bounds. This notion of typicality was then used in the analysis of multiparty state merging and assisted entanglement distillation protocols. In the second work, the authors provided a similar statement for the one-shot case. Specifically, for a given multiparty state, they conjectured the existence of a state that is close to the initial state but has a min-entropy bounded by the smoothed min-entropy of the initial state for all marginals. In a follow up paper we will try to understand the relationship between these various notions of quantum joint typicality and whether Sen's results can be extended to prove the other notions or to realize the applications they are designed for.

Also, as noted in the corresponding section, our protocol for the partial decode-forward bound is not a straightforward generalization of the classical protocol in~\cite{elgamal2011network}. Our algorithm involves a joint measurement of all the transmitted blocks instead of performing a backward decoding followed by a forward decoding as in the classical case. The problem arises from the fact that the classical protocol makes both multiple measurements on a single system and also \emph{intermediate} measurements on other systems. Hence, a direct application of our packing lemma has to combine both the multiple measurements on the same system and the intermediate measurements into one joint measurement. This is done by applying Sen's one-shot joint typicality lemma on \emph{all} of the receiver's systems. This results in a set of inequalities for the rate region that has to be simplified to obtain the desired bound. This is a step that might be necessary in other applications of our packing lemma. 

There are still several interesting questions that remain open regarding quantum relay channels. The most obvious one is proving converses for the given achievability lower bounds. There are known converses for special classical relay channels, and it would be interesting to extend them to the quantum case as we did for semideterministic relay channels. Another, albeit less trivial, direction is to prove a quantum equivalent of the compress-forward lower bound~\cite{elgamal2011network}. We might need to analyze this in the entanglement assisted case since it is only then that a single-letter quantum rate-distortion theorem is known~\cite{datta2013quantum}. Another idea is to study networks of relay channels, where the relays are operating in series or in parallel. Some preliminary work was done in~\cite{jin2012lower}, and the most general notion of this in the classical literature is a multicast network~\cite{elgamal2011network}. Lastly, relay channels with feedback would also be interesting to investigate.

\newpage 
\noindent \textbf{Acknowledgements.} We thank Pranab Sen for interesting discussions and for sharing his draft~\cite{senInPrep} with us. We would also like to thank Mario Berta, Philippe Faist, and Mark Wilde for inspiring discussions. PH was supported by AFOSR (FA9550-16-1-0082), CIFAR and the Simons Foundation. HG was supported in part by NSF grant PHY-1720397. MW acknowledges financial support by the NWO through Veni grant no.~680-47-459. DD is supported by the Stanford Graduate Fellowship and the National Defense Science and Engineering Graduate Fellowship. DD would like to thank God for all of His provisions.

\bibliography{relay}

\appendix

\section{Proof of Cutset Bound}\label{app:cutset}
We give a proof of~\cref{prp:cutset}, essentially identical to that of~\cite{elgamal2011network}:
\begin{proof}
  Consider an $(n, 2^{n R})$ code for $\mathcal{N}_{X_1 X_2 \to B_2 B_3}$. Suppose we have a uniform distribution over the message set $M$, and denote the final classical system obtained by Bob from the POVM measurement by $\hat M$. By the classical Fano's inequality,
  \begin{align*}
    n R = H(M) = I(M; \hat M) + H(M | \hat M) \le I(M; \hat M) + n \delta(n),
  \end{align*}
  where $\delta(n)$ satisfies $\lim_{n \to \infty} \delta(n) = 0$ if the decoding error is to vanish in asymptotic limit.

  We denote by $(X_1)_j, (X_2)_j, (B_2)_j, (B_3)_j$ the respective classical and quantum systems induced by our protocol. We argue
  \begin{align*}
    I(M; \hat M) \le I(M; B_3^n) & = \sum_{j=1}^{n} I(M; (B_3)_j | B_3^{j-1})\\
    & \le \sum_{j=1}^{n} I(M B_3^{j-1} ; (B_3)_j) \\
    & \le \sum_{j=1}^{n} I( (X_1)_j (X_2)_j M B_3^{j-1} ; (B_3)_j) \\
    & = \sum_{j=1}^{n} I((X_1)_j (X_2)_j; (B_3)_j).
  \end{align*}
  The last step follows from the i.i.d.\ nature of the $n$ channel uses and the channel is classical-quantum. More explicitly, we can write out the overall state as the protocol progresses, and since the input to the channel on each round is classical, it is not difficult to see that given $(X_1)_j (X_2)_j$, $(B_2)_j (B_3)_j$ is in tensor product with the other systems. This would not hold if the channel takes quantum inputs, for which we would expect an upper bound that involves regularization. Now, similarly,
  \begin{align*}
    & I(M; \hat M) \le I(M; B_3^n) \\
    & \le I(M; B_2^n B_3^n) \\
    & = \sum_{j=1}^{n} I(M; (B_2)_j (B_3)_j | B_2^{j-1} B_3^{j-1}) \\
    & = \sum_{j=1}^{n} I(M; (B_2)_j (B_3)_j | B_2^{j-1} B_3^{j-1} (X_2)_j) \\
    & \le \sum_{j=1}^{n} I(MB_2^{j-1}B_3^{j-1}; (B_2)_j (B_3)_j |   (X_2)_j) \\
    & \le \sum_{j=1}^{n} I( (X_1)_j MB_2^{j-1}B_3^{j-1}; (B_2)_j (B_3)_j |   (X_2)_j) \\
    & = \sum_{j=1}^{n} I( (X_1)_j ; (B_2)_j (B_3)_j |   (X_2)_j),
  \end{align*}
  where the second equality follows since given $B_2^{j-1}$, one can obtain $(X_2)_j$ by a series of $\mathcal{R}$ operations (Note that $(B_2)_0(B_2')_0$ is a trivial system and thus independent of the code.).

  Define the state
  \begin{align*}
    \sigma_{Q X_1 X_2 B_2 B_3} \equiv \frac{1}{n} \sum_{q=1}^{n} \state{q}_Q \ot \sigma^{(q)}_{X_1 X_2 B_2 B_3},
  \end{align*}
  where $\sigma^{(q)}$ is the classical-quantum state on the $q$\textsuperscript{th} round of the protocol, that is, the state on the system $(X_1)_q (X_2)_q (B_2)_q (B_3)_q$. Now, $I(B_2 B_3 ; Q | X_1 X_2)_\sigma =0$, so
  \begin{align*}
    \sum_{j=1}^{n} I((X_1)_j (X_2)_j; (B_3)_j) & = b I(X_1 X_2 ; B_3 |Q )_\sigma \\
    & \le n I(X_1 X_2 Q; B_3)_\sigma \\
    & = n I(X_1 X_2; B_3)_\sigma
  \end{align*}
  and similarly
  \begin{align*}
    \sum_{j=1}^{n} I( (X_1)_j ; (B_2)_j (B_3)_j |   (X_2)_j) & = n I(X_1 ; B_2 B_3 | X_2 Q)_\sigma \\
    & \le n I(X_1 Q; B_2 B_3 | X_2)_\sigma \\
    & =  n I(X_1; B_2 B_3 | X_2)_\sigma.
  \end{align*}
  Hence,
  \begin{align*}
    R \le \min\{I(X_1 X_2; B_3)_\sigma, I(X_1; B_2 B_3 | X_2)_\sigma \} + \delta(n).
  \end{align*}
  Now, $\sigma_{X_1 X_2 B_2 B_3}$ is simply a uniform average of all the classical-quantum states from each round of the protocol, it is also a possible classical-quantum state induced by $\mathcal{N}_{X_1 X_2 B_2 B_3}$ acting on some classical input distribution $p_{X_1 X_2}$. In particular, $R$ is therefore upper bounded by the input distribution which maximizes the quantity on the right-hand side:
  \begin{align*}
    R \le \max_{p_{X_1X_2}}\min\{I(X_1 X_2; B_3), I(X_1; B_2 B_3 | X_2)\} + \delta(n).
  \end{align*}
  Taking the $n \to \infty$ limit completes the proof.
\end{proof}

\end{document}